%% file: simoncombi.tex
\documentclass{llncs}
\usepackage{amsmath}
\usepackage{amsfonts}
\usepackage{amssymb}
\usepackage{comment}
\usepackage{algorithm}
\usepackage{algpseudocode}
\usepackage[utf8]{inputenc}

\usepackage{tikz}
\usepackage{pgfplots}
\usepackage{bbm}
\usepackage{multirow}
\usepackage{enumerate}

\usepackage{xcolor}
\usepackage{xifthen}
\usepackage{booktabs}

\usepackage[matrix,frame,arrow]{xypic}

\input{Qcircuit}

\newcommand{\zo}{\{0,1\}}
\newcommand{\good}{{\sf good}}
\newcommand{\bad}{{\sf bad}}

\newcommand{\AlgQone}{{\sf Alg\text{-}ExpQ1}}
\newcommand{\AlgQtwo}{{\sf Alg\text{-}PolyQ2}}
\newcommand{\SimQone}{{\sf SimQ1}}

\newcommand{\bigO}[1]{ \mathcal{O} \left(#1 \right)}
\newcommand{\bigOt}[1]{ \widetilde{\mathcal{O}} \left(#1 \right)}
\newcommand{\FX}{\mathrm{FX}}

\usetikzlibrary{shapes.misc,positioning}
\tikzstyle{block}=[draw,minimum size=2em]
\tikzstyle{l}=[minimum size=1.5em]
\tikzstyle{xor}=[draw, minimum size=0.8em,append after command={[shorten >=\pgflinewidth, shorten <=\pgflinewidth,] 
        (\tikzlastnode.north) edge (\tikzlastnode.south)
        (\tikzlastnode.east) edge (\tikzlastnode.west)
}]
\tikzstyle{sponge}=[rectangle, rounded corners=.25cm, minimum width=.5cm, minimum height=1.8cm, draw]

\newcommand{\fullversion}[1]{#1}
\newcommand{\acversion}[1]{}

\title{Quantum Attacks without Superposition Queries: the Offline Simon's Algorithm}

\author{Xavier Bonnetain\inst{1,3} \and Akinori Hosoyamada\inst{2,4} \and Mar\'ia Naya-Plasencia\inst{1} \and Yu Sasaki\inst{2} \and Andr\'e Schrottenloher\inst{1}}
\institute{Inria, France\\
  \texttt{\{xavier.bonnetain,maria.naya\_plasencia,andre.schrottenloher\}@inria.fr}
  \and NTT Secure Platform Laboratories, Tokyo, Japan\\
  \texttt{\{hosoyamada.akinori,sasaki.yu\}@lab.ntt.co.jp}
  \and Sorbonne Universit\'e, Coll\`ege Doctoral, F-75005 Paris, France
  \and Nagoya University, Nagoya, Japan}

\begin{document}
\maketitle
\renewcommand{\labelitemi}{$\bullet$}
\setcounter{footnote}{0}

\begin{abstract}
In symmetric cryptanalysis, the model of superposition queries has led to surprising results, with many constructions being broken in polynomial time thanks to Simon's period-finding algorithm. But the practical implications of these attacks remain blurry. In contrast, the results obtained so far for a quantum adversary making classical queries only are less impressive.

In this paper, we introduce a new quantum algorithm which uses Simon's subroutines in a novel way. We manage to leverage the algebraic structure of cryptosystems in the context of a quantum attacker limited to classical queries and offline quantum computations. We obtain improved quantum-time/classical-data tradeoffs with respect to the current literature, while using only as much hardware requirements (quantum and classical) as a standard exhaustive search with Grover's algorithm. In particular, we are able to break the Even-Mansour construction in quantum time $\tilde{O}(2^{n/3})$, with $O(2^{n/3})$ classical queries and $O(n^2)$ qubits only. In addition, we improve some previous superposition attacks by reducing the 
data complexity from exponential to polynomial, with the same time complexity.

Our approach can be seen in two complementary ways: \emph{reusing} superposition queries during the iteration of a search using Grover's algorithm, or alternatively, removing the memory requirement in some quantum attacks based on a collision search, thanks to their algebraic structure.

We provide a list of cryptographic applications, including the Even-Mansour construction, the FX construction, some Sponge authenticated modes of encryption, and many more.

\end{abstract}


\section{Introduction}
\input{introduction.tex}

\section{Preliminaries}
\label{sec:prelim}
\input{preliminaries.tex}

\section{Simon's Algorithm with Asymmetric Queries}\label{sec:algorithm}
\input{algorithm.tex}

\section{Q2 Attacks on Symmetric Schemes with Reduced Query Complexity}\label{sec:q2}
\input{q2applications.tex}

\section{Q1 Attacks on Symmetric Schemes}\label{sec:q1}
\input{q1applications.tex}

\section{Discussion}\label{sec:discussion}
\input{discussion.tex}

\section{Conclusion}\label{sec:conclusion}

In this paper, we have introduced a new quantum algorithm, in which we make use of Simon's algorithm in an \emph{offline} way. The idea of making $\mathsf{poly}(n)$ superposition queries to the oracle (with, as input, a uniform superposition), storing them as some compressed database on $n^2$ qubits, and reusing them during the iterations of a Grover search, yielded surprising results. This idea, initially targeting the query complexity of some Q2 attacks on cryptographic schemes, enabled us to find new quantum-time/classical-data tradeoffs.
Our result has three consequences, each of which answers a long-standing question in post-quantum cryptography.
\paragraph{Simon's Algorithm can be Used in An Offline Setting.}
We provided the first example of use of Simon's algorithm (or more precisely, its core idea) in an offline setting, without quantum oracle queries.
\paragraph{Improving More than the Time Complexity.}
Consider the example of our attack on the Even-Mansour construction in quantum time $\bigOt{2^{n/3}}$ and classical queries $\bigO{2^{n/3}}$. With the same number of queries, the classical attack requires $\bigO{2^{2n/3}}$ time \emph{and} $\bigO{2^{n/3}}$ classical memory to store the queries. In our attack, we do not need this storage. To the best of our knowledge, this is the first time that a quantum Q1 attack provides a quadratic speedup while the needs of hardware are also reduced.
\paragraph{Q2 Attacks Make a Difference.}
Schemes which do not have an attack in the superposition model cannot be attacked by our algorithm. We showed that their algebraic structure, which makes the superposition attack possible, indeed made a practical difference when it came to Q1 attacks. We believe that this question needs further investigation.

\subsubsection*{Acknowledgements. }
The authors thank L\'eo Perrin for proofreading this article and Elena Kirshanova for helpful remarks.
This project has received funding from the European Research Council (ERC) under the European Union's Horizon 2020 research and innovation programme (grant agreement $\mbox{n}^o$~714294 - acronym QUASYModo).

\bibliography{biblio}
\bibliographystyle{splncs03}

\fullversion{
\appendix
\input{erroranalyses}

\section{Adaptive Attacks and Non-Adaptive Attacks}\label{sec:adaptive}
Our attack on the FX construction in the Q2 model is \emph{non-adaptive} since the online queries required to make the state $\ket{\psi_g}$ is just the uniform superposition of plaintexts.
On the other hand, the previous Q2 attack on the construction by Leander and May is \emph{adaptive} since quantum queries made to keyed online oracles are changed depending on the results for previous quantum queries.
For Q1 attacks, both of existing quantum attacks and our attacks are non-adaptive.
}
 
\end{document}

%% file: introduction.tex

Ever since Shor~\cite{DBLP:conf/focs/Shor94} introduced his celebrated quantum polynomial-time algorithm for solving factorization and Discrete Logarithms, both problems believed to be classically intractable, post-quantum cryptography has become a subject of wide interest. Indeed, the security of classical cryptosystems relies on computational assumptions, which until recently, were made with respect to classical adversaries; if quantum adversaries are to be taken into account, the landscape of security is bound to change dramatically. 

While it is difficult to assert the precise power of quantum computers, which are yet to come, it is still possible to study quantum algorithms for cryptographic problems, and to estimate the computational cost of solving these problems for a quantum adversary. 
The ongoing project by NIST~\cite{nistcall} for post-quantum \emph{asymmetric} schemes aims to replace the current mostly used ones by new standards.

In \emph{symmetric} cryptography, the impact of quantum computing seems, at first sight, much more limited. 
This is because the security of most of symmetric-key schemes is not predicated on structured problems. Symmetric-key schemes are required to be computed extremely efficiently, and designers must avoid such computationally expensive operations.
Grover's quantum search algorithm~\cite{DBLP:conf/stoc/Grover96}, another cornerstone of quantum computing, speeds up by a quadratic factor exhaustive search procedures. This has led to the common saying that ``doubling the key sizes'' should ensure a similar level of post-quantum security.

However, the actual post-quantum security of symmetric-key schemes requires more delicate treatment. Recovering the secret key via exhaustive search is only one of all the possible approaches.
The report of the National Academy of Sciences on the advent of quantum computing~\cite{NAP25196} also states that ``it is possible that there is some currently unknown clever quantum attack'' that would perform much better than Grover's algorithm. Indeed, cryptographers are making significant progress on quantum attackers with \emph{superposition queries}, which break many symmetric-key schemes in polynomial time.

\paragraph{Quantum Generic Attacks in Q1 and Q2 Models.}
Quantum attacks can be mainly classified into two types~\cite{DBLP:phd/dnb/Gagliardoni17,DBLP:journals/tosc/KaplanLLN16,DBLP:conf/ctrsa/HosoyamadaS18}, Q1 model and Q2 model, assuming different abilities for the attacker. In the Q1 model, attackers have an access to a quantum computer to perform any offline computation, while they are only allowed to make online queries in a classical manner. In the Q2 model, besides the offline quantum computation, attackers are allowed to make superposition queries to a quantum cryptographic oracle. Here, we briefly review previous results in these models to introduce the context of our results.

The Q2 model is particularly interesting as it yields some attacks with a very low cost. 
Kuwakado and Morii~\cite{DBLP:conf/isit/KuwakadoM10,DBLP:conf/isita/KuwakadoM12} showed that the Even-Mansour cipher and the three-round Feistel networks, classically proven secure if their underlying building blocks are ideal, were broken in polynomial time. This exponential speedup, the first concerning symmetric cryptography, was obtained thanks to Simon's algorithm~\cite{DBLP:conf/focs/Simon94} for recovering a Boolean hidden shift.
Later on, more results have been obtained in this setting, with more generic constructions broken~\cite{DBLP:conf/crypto/KaplanLLN16,DBLP:conf/asiacrypt/Leander017}, and an exponential acceleration of \emph{slide attacks}, which target ciphers with a self-similar structure. Versions of these attacks~\cite{DBLP:conf/asiacrypt/BonnetainN18} for constructions with modular additions use Kuperberg's algorithm~\cite{DBLP:journals/siamcomp/Kuperberg05}, allowing a better than quadratic speed-up. 
All these attacks, however, run in the model of \emph{superposition queries}, which models a quantum adversary having some inherently quantum access to the primitives attacked. As such, they do not give any improvement when the adversary only has classical access.

Stated differently, the attacks in the Q1 model are particularly relevant due to their impact on current data communication technology. 
However, the quantum algorithms that have been exploited for building attacks in the Q1 model are very limited and have not allowed more than a quadratic speed-up. The most used algortihm is the simple quantum exhaustive search with Grover's algorithm. A possible direction is the collision finding algorithm that is often said to achieve ``$2^{n/3}$ complexity" versus $2^{n/2}$ classically. However, even in this direction, there are several debatable points; basic quantum algorithms for finding collisions have massive quantum hardware requirements~\cite{DBLP:conf/latin/BrassardHT98}. There is a quantum-hardware-friendly variant~\cite{DBLP:conf/asiacrypt/ChaillouxNS17}, but then the time complexity becomes suboptimal. 

In summary, attacks using Simon's algorithm could achieve a very low complexity but could only be applied in the Q2 model, a very strong model.
In contrast, attacks in the Q1 model are practically more relevant, but for now the obtained speed-ups were not surprising.

Another model to consider when designing quantum attacks is whether the attacker has or not a big amount of quantum memory available.
Small quantum computers seem like the most plausible scenario, and therefore attacks needing a polynomial amount of qubits are more practically relevant.
Therefore, the most \emph{realistic} scenario is Q1 with small quantum memory.

\paragraph{Our Main Contribution.}
The breakthrough we present in this paper is the first application of Simon's algorithm~\cite{DBLP:conf/focs/Simon94} in the Q1 model, which requires significantly less than $\bigO{2^{n/2}}$ classical queries and offline quantum computations, only with $\mathsf{poly}(n)$ qubits, and no qRAM access (where $n$ is the size of the secret).
Namely, we remove the superposition queries in previous attacks. The new idea can be applied to a long list of ciphers and modes of operation. Let us illustrate the impact of our attacks by focusing on two applications:

The first application is the key recovery on the Even-Mansour construction, which is one of the simplest attacks using Simon's algorithm.
Besides the polynomial time attacks in the Q2 model, Kuwakado and Morii also developed an attack in the Q1 model with $\bigO{2^{n/3}}$ classical queries, quantum computations, qubits, and classical memory~\cite{DBLP:conf/isita/KuwakadoM12}.
The extension of this Q1 attack by Hosoyamada and Sasaki~\cite{DBLP:conf/ctrsa/HosoyamadaS18} recovers the key with $\bigO{2^{3n/7}}$ classical queries, $\bigO{2^{3n/7}}$ quantum computations, polynomially many qubits and $\bigO{2^{n/7}}$ classical memory (to balance classical queries and quantum computations). Our attack in the Q1 model only uses polynomially many qubits, yet only requires $\bigO{2^{n/3}}$ classical queries, $\bigO{n^3 2^{n/3}}$ quantum computations and $\mathsf{poly}(n)$ classical memory.

The second application is the key recovery on the FX-construction $\mathrm{FX}_{k,k_{in},k_{out}}$, which computes a ciphertext $c$ from a plaintext $p$ by $c \leftarrow E_{k}(p \oplus k_{in}) \oplus k_{out}$, where $E$ is a block cipher, $k$ is an $m$-bit key and $k_{in},k_{out}$ are two $n$-bit keys. Leander and May proposed an attack in the Q2 model with $\bigO{n2^{m/2}}$ superposition queries, $\bigO{n^3 2^{m/2}}$ quantum computations, $\mathsf{poly}(n)$ qubits and $\mathsf{poly}(n)$ classical memory \cite{DBLP:conf/asiacrypt/Leander017}.
\footnote{
Here we are assuming that $m$ is in $\bigO{n}$, which is the case for usual block ciphers.
}
They combined Simon's algorithm and Grover's algorithm in a clever way, while it became inevitable to make queries in an adaptive manner. For the Q1 model, the meet-in-the-middle attack \cite{DBLP:conf/ctrsa/HosoyamadaS18} can recover the key with $\bigO{2^{3(m+n)/7}}$ complexities.
Our results can improve the previous attacks in two directions. One is to reduce the amount of superposition queries in the Q2 model to the polynomial order and convert the adaptive attack to a non-adaptive one. The other is to completely remove the superposition queries.
The comparison of previous quantum attacks and our attacks on Even-Mansour and the FX construction is shown in Table~\ref{Tbl:comparison}. Other interesting complexity trade-offs are possible, as shown in detail in sections~\ref{sec:q1} and~\ref{sec:q2}.
\begin{table}[!htb]
\begin{center}
\caption{Previous and New Quantum Attacks on Even-Mansour and FX, assuming that $m = \bigO{n}$.} \label{Tbl:comparison}
\begin{tabular}{@{}c|@{\ }cccccc@{}} 
\toprule
Target & {Model} & Queries & Time & Q-memory & C-memory & Reference \\ 
\midrule
             & Q2 &  $\bigO{n}$ &  $\bigO{n^{3}}$ &  $\bigO{n}$   & $\bigO{n^2}$ & \cite{DBLP:conf/isita/KuwakadoM12} \\
             & Q1 &  $\bigO{2^{n/3}}$ &  $\bigO{2^{n/3}}$ &  $\bigO{2^{n/3}}$   & $\bigO{2^{n/3}}$ & \cite{DBLP:conf/isita/KuwakadoM12} \\
EM & Q1 & $\bigO{2^{3n/7}}$  & $\bigO{2^{3n/7}}$   & $\bigO{n}$    & $\bigO{2^{n/7}}$  & \cite{DBLP:conf/ctrsa/HosoyamadaS18} \\
             & Q1  & $\bigO{2^{n/3}}$  & $\bigO{n^3 2^{n/3}}$  &  $\bigO{n^2}$  & $\bigO{n}$ & Section~\ref{sec:q1}  \\ 
\midrule
             & Q2  & $\bigO{n2^{m/2}}$  &  $\bigO{n^{3}2^{m/2}}$   & $\bigO{n^2}$  & 0 & \cite{DBLP:conf/asiacrypt/Leander017} \\
  FX          & Q2  & $\bigO{n}$ & $\bigO{n^{3} 2^{m/2}}$  & $\bigO{n^2}$ & $\bigO{n}$ & Section~\ref{sec:q2}   \\
             & Q1  & $\bigO{2^{3(m+n)/7}}$  & $\bigO{2^{3(m+n)/7}}$ &  $\bigO{n}$   & $\bigO{2^{(m+n)/7}}$   & \cite{DBLP:conf/ctrsa/HosoyamadaS18} \\
              & Q1  & $\bigO{2^{(m+n)/3}}$  & $\bigO{n^3 2^{(m+n)/3}}$  & $\bigO{n^2}$  & $\bigO{n}$ & Section~\ref{sec:q1}   \\ 
\bottomrule
\end{tabular}
\end{center}
\end{table}

\paragraph{Our New Observation.}
Here we describe our new algorithm used in the Q1 model with the Even-Mansour construction as an example.
Recall that the encryption $E_{k_1,k_2}$ of the Even-Mansour construction is defined as $E_{k_1,k_2}(x) = P(x \oplus k_1) \oplus k_2$, where $P$ is a public permutation and $k_1,k_2 \in \{0,1\}^n$ are the secret keys.
Roughly speaking, our attack guesses $(2n/3)$-bit of $k_1$ (denoted by $k^{(2)}_1$ in Fig.~\ref{fig:Q1EMIntro}) by using the Grover search, and checks if the guess is correct by applying Simon's algorithm to the remaining $(n/3)$-bit of $k_1$ (denoted by $k^{(1)}_1$ in Fig.~\ref{fig:Q1EMIntro}).
If we were in the Q2 model, we could recover $k_1$ by using the technique by Leander and May~\cite{DBLP:conf/asiacrypt/Leander017} in time $\tilde{\mathcal{O}}(2^{n/3})$.
However, their technique is not applicable in the Q1 setting since quantum queries are required.

Our core observation that realizes the above idea in the Q1 model is that, we can judge whether a function $f \oplus g$ has a period (i.e., we can apply Simon's algorithm) without any quantum query to $g$, if we have the quantum state $\ket{\psi_g} := \left( \sum_{x} \ket{x}\ket{g(x)} \right)^{\otimes cn}$ ($c$ is a small constant):
If we have the quantum state $\ket{\psi_g}$, then we can make the quantum state $\ket{\psi_{f\oplus g}} := \left( \sum_{x} \ket{x}\ket{(f \oplus g)(x)} \right)^{\otimes cn}$ by making $\bigO{n}$ quantum queries to $f$.
Once we obtain $\ket{\psi_{f \oplus g}}$, by applying the Hadamard operation $H^{\otimes n}$ to each $\ket{x}$ register, we obtain the quantum state
$$\left( \sum_{x_1,u_1} (-1)^{u_1 \cdot x_1} \ket{u_{1}}\ket{(f \oplus g)(x_{1})} \right) \otimes \cdots \otimes  \left( \sum_{x_{cn},u_{cn}} (-1)^{u_{cn} \cdot x_{cn}} \ket{u_{cn}}\ket{(f \oplus g)(x_{cn})} \right)$$
Then, roughly speaking, $\mathrm{dim}(\mathrm{Span}(u_1,\dots,u_{cn})) < n$ always holds if $f \oplus g$ has a secret period $s$, while $\mathrm{dim}(\mathrm{Span}(u_1,\dots,u_{cn})) = n$ holds with a high probability if $f \oplus g$ does not have any period.
Since the dimension of the vector space can be computed in time $\bigO{n^3}$, we can judge if $f \oplus g$ has a period in time $\bigO{n^3}$.
Note that we can reconstruct the quantum data $\ket{\psi_g}$ after judging whether $(f \oplus g)$ has a period (with some errors) by appropriately performing uncomputations, which help us use these procedures as a subroutine without measurement in other quantum algorithms.

For the Even-Mansour construction, we set $g : \{0,1\}^{n/3} \rightarrow \{0,1\}^n$ by $g(x) := E_{k_1,k_2}(x \| 0^{2n/3})$.
Then we can make the quantum state $\ket{\psi_g}$ by classically querying $x$ to $g$ for \emph{all} $x \in \{0,1\}^{n/3}$, which requires $2^{n/3}$ classical queries.
After obtaining the state $\ket{\psi_g}$, we guess $k^{(2)}_1$.
Suppose that here our guess is $k' \in \{0,1\}^{2n/3}$.
We define $f_{k'} : \{0,1\}^{n/3} \rightarrow \{0,1\}^n$ by $f_{k'}(x) := P(x \| k')$. 
Then, roughly speaking, our guess is correct if and only if the function $f_{k'} \oplus g$ has a period $k^{(1)}_1$.
Thus we can judge whether the guess is correct without quantum queries to $g$, by using our technique described above.
Since $k^{(2)}_1$ can be guessed in time $\tilde{\mathcal{O}}(2^{n/3})$ by using the Grover search, we can recover the keys by making $\mathcal{O}(2^{n/3})$ classical queries and $\tilde{\mathcal{O}}(2^{n/3})$ offline quantum computations.
\usetikzlibrary{arrows}
\newcommand{\midsection}{\tikz \draw[-] (-0.1,-0.1) -- (0.1,0.1);}

\begin{figure}
\centering
\begin{tikzpicture}[node distance=2em]
  \draw
  node at (0,0)[block,name=p,minimum size=3em] {$P$}
  node at (p.west) [name=pleft] {}
  node [name=pup, above of=pleft, node distance=1em] {}
  node [xor, name=xorup,circle,left of=pup,node distance=3em] {}
  node [l, above of=xorup,name=k12, above] {$k_1^{(2)}$}
  node [name=pdown, below of=pleft, node distance=1em] {}
  node [xor, name=xordown, circle, left of=pdown, node distance=3em] {}
  node [l, below of=xordown, name=k11,below] {$k_1^{(1)}$}
  node [left of=xorup, node distance = 7em,name=startup] {}
  node [left of=xordown, node distance = 7em,name=startdown] {};
  \draw[->] (xorup) -- (pup.center);
  \draw[->] (xordown) -- (pdown.center);
  \draw[->] (k12) -- (xorup);
  \draw[->] (k11) -- (xordown);
  \draw (startup) -- node {\midsection} node[above,node distance=1.5em] {$\frac{2n}{3}$} (xorup);
  \draw (startdown) -- node {\midsection} node[below,node distance=2em] {$\frac{n}{3}$} (xordown);
  \draw
  node [xor, circle,name=xorend, right of=p, node distance=3em] {}
  node [l, above of=xorend,name=k2,above] {$k_2$}
  node [right of=xorend, node distance=7em, name=end] {};
  \draw (p) -- (xorend);
  \draw[->] (k2) -- (xorend);
  \draw[->] (xorend) -- node {\midsection} node[name=n,above, node distance=1.5em] {$n$} (end);
  
  \draw 
  node [block, dotted,name=grover, above of=n, node distance=3.5em,right=-2em] {Grover search space}
  node [block, dotted,name=simon, below of=n, node distance=4.5em] {Apply Simon's algorithm};
  \draw[dotted, ->] (grover.west|-k12) -- (k12);
  \draw[dotted,->] (simon.west|-k11) -- (k11);
  
 \end{tikzpicture}
\caption{Idea of our Q1 attack on the Even-Mansour construction.}
\label{fig:Q1EMIntro}
\end{figure}
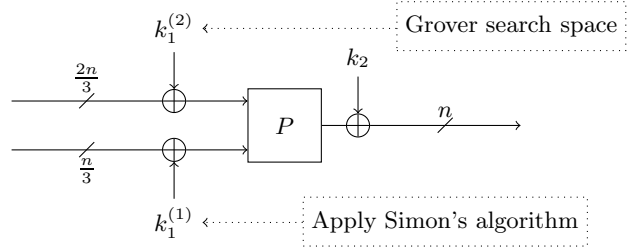

We will show how we can similarly attack the FX construction in the Q1 model, by guessing additional key bits (see Fig.~\ref{fig:Q1FXIntro}).

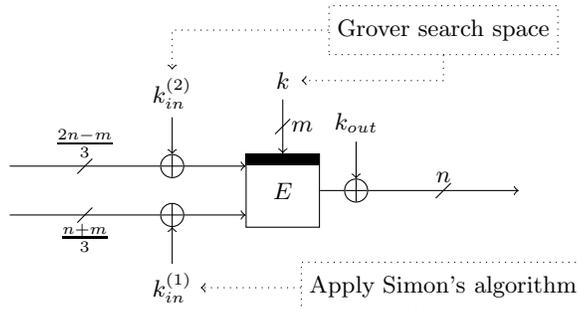
\begin{figure}[htb]
\centering
\begin{tikzpicture}[node distance=2em]
  \draw
  node at (0,0)[block,name=p,minimum size=3em] {$E$}
  node at (p.west) [name=pleft] {}
  node [name=pup, above of=pleft, node distance=1em] {}
  node [xor, name=xorup,circle,left of=pup,node distance=3em] {}
  node [l, above of=xorup,name=k12,above] {$k_{in}^{(2)}$}
  node [name=pdown, below of=pleft, node distance=1em] {}
  node [xor, name=xordown, circle, left of=pdown, node distance=3em] {}
  node [l, below of=xordown, name=k11, below] {$k_{in}^{(1)}$}
  node [left of=xorup, node distance = 7em,name=startup] {}
  node [left of=xordown, node distance = 7em,name=startdown] {};
  \draw[->] (xorup) -- (pup.center);
  \draw[->] (k12) -- (xorup);
  \draw[->] (xordown) -- (pdown.center);
  \draw[->] (k11) -- (xordown);
  \draw (startup) -- node {\midsection} node[above,node distance=1.5em] {$\frac{2n-m}{3}$} (xorup);
  \draw (startdown) -- node {\midsection} node[below,node distance=2em] {$\frac{n+m}{3}$} (xordown);
  \draw
  node [xor, circle,name=xorend, right of=p, node distance=3em] {}
  node [l, above of=xorend,name=k2,above] {$k_{out}$}
  node [right of=xorend, node distance=7em, name=end] {};
  \draw[->] (p) -- (xorend)-- node {\midsection} node[name=n,above, node distance=1.5em] {$n$} (end);
  \draw[->] (k2) -- (xorend);
  \draw 
  node [block, dotted,name=grover, above of=n, node distance=6em] {Grover search space}
  node [block, dotted,name=simon, below of=n, node distance=4.5em] {Apply Simon's algorithm};

  \draw[dotted,->] (simon.west|-k11) -- (k11);
  \draw[fill=black,draw opacity=0] (p.north west) rectangle ([yshift=-0.15cm]p.north east);
  \draw  
  node [l,name=k,above of=p, node distance=4.5em] {$k$};
  \draw[->] (k) -- node {\midsection} node [name=m,right,node distance=1.5em] {$m$} (p);
  \draw  node [name=tmp,above of=k] {};
  \draw[dotted, ->] (grover) -| (k12);
  \draw[dotted, ->] (grover) |- (k);
 \end{tikzpicture}
\caption{Idea of our Q1 attack on the FX construction.}
\label{fig:Q1FXIntro}
\end{figure}

Moreover, our attack idea in the Q1 model can also be used to reduce the number of quantum queries of attacks in the Q2 model.
The Leander and May's attack on the FX construction in the Q2 model~\cite{DBLP:conf/asiacrypt/Leander017} guesses the $m$-bit key $k$ of the FX construction $\mathrm{FX}_{k,k_{in},k_{out}}$ and checks whether the guess is correct by using Simon's algorithm, which requires $\mathcal{O}(2^{m/2})$ online quantum queries and $\tilde{\mathcal{O}}(2^{m/2})$ offline quantum computations.
Roughly speaking, the guess $k'$ for the key $k$ is correct if and only if $(f_{k'} \oplus g)(x)$ has the secret period $k_{in}$, where $f_{k'}(x) = E_{k'}(x)$ and $g(x) = \mathrm{FX}_{k,k_{in},k_{out}}(x)$.
In the Q2 model, we can make the quantum state $\ket{\psi_g} = \left(\sum_x \ket{x}\ket{g(x)}\right)^{\otimes cn}$ by making $\bigO{n}$ quantum queries to $g$.
Thus, by our new attack idea described above, we can break the FX construction with $\bigO{n}$ online quantum queries and $\tilde{\mathcal{O}}(2^{m/2})$ offline quantum computations, which exponentially improves the attack by Leander and May from the viewpoint of \emph{quantum query} complexity.

This exponential improvement on the quantum query complexity is due to the separation of offline queries and online computations:
In the previous attack on the FX construction in the Q2 model by Leander and May, we have to do online queries and offline computations alternately in each iteration of the Grover search.
Thus the number of online quantum queries becomes exponential in the previous attack. 
On the other hand, in our new attack, the online queries (i.e., the procedures to make the quantum state $\ket{\psi_g}$) are completely separated from offline computations.
This enables us to decrease the number of quantum queries exponentially, while we still need exponentially many offline computations.

\subsubsection{Paper organization.}
Section~\ref{sec:prelim} gives preliminaries.
Section~\ref{sec:algorithm} describes our main algorithms.
Section~\ref{sec:q2} shows applications of our algorithms in the Q2 model.
Section~\ref{sec:q1} shows applications of our algorithms in the Q1 model.
Section~\ref{sec:discussion} discusses further applications of our algorithm.
Section~\ref{sec:conclusion} concludes the paper.

%% file: preliminaries.tex

In this section, we introduce some quantum computing notions and review Simon's and Grover's algorithms. We refer to~\cite{nielsen2002quantum} for a broader presentation.

\subsection{The Quantum Circuit Model}

It has become standard in the cryptographic literature to write quantum algorithms in the circuit model, which is universal for quantum computing. We only consider the logical level of quantum circuits, with logical qubits, not their implementation level (which requires physical qubits, quantum error-correction, \emph{etc}). Although it is difficult to estimate the cost of a physical implementation which does not yet exist, we can compare security levels as quantum operation counts in this model. For example, Grover search of the secret key for AES-128 is known to require approximately $2^{64}$ quantum evaluations of the cipher, and $2^{84}$ quantum operations~\cite{DBLP:conf/pqcrypto/GrasslLRS16}.

\paragraph{Qubits and Operations.}
A quantum circuit represents a sequence of quantum operations, denoted as \emph{quantum gates}, applied to a set of \emph{qubits}. An individual qubit is a quantum object whose state is an element of a two-dimensional Hilbert space, with basis $\ket{0}, \ket{1}$ (analogs of the classical logical $0$ and $1$). Hence, the state is described as a linear combination of $\ket{0}, \ket{1}$ with complex coefficients (a \emph{superposition}). We add to this a normalization condition: $\alpha \ket{0} + \beta \ket{1}$ is such that $|\alpha|^2 + |\beta|^2 = 1$. When it is clear from context, we dismiss common normalization factors.

When $n$ qubits are given, the computational basis has $2^n$ vectors, which are all $n$-bit strings. The qubits start in a state $\ket{0}$, for example a fixed spin or polarization. The sequence of quantum gates that is applied modifies the superposition, thanks to constructive and destructive interferences. In the end, we measure the system, and obtain some $n$-bit vector in the computational basis, which we expect to hold a meaningful result.

All computations are (linear) unitary operators of the Hilbert space, and as such, are reversible (this holds for the individual gates, but also for the whole circuit). In general, any classical computation can be made reversible (and so, implemented as a quantum circuit) provided that one uses sufficiently many \emph{ancilla qubits} (which start in the state $\ket{0}$ and are brought back to $\ket{0}$ after the computation). Generally, on input $\ket{x}$, we can perform some computation, copy the result to an output register using CNOT gates, and uncompute (perform backwards the same operations) to restore the initial state of the ancilla qubits. Uncomputing a unitary $U$ corresponds to applying its adjoint operator $U^*$.

By the principle of \emph{deferred measurements}, any measure that occurs inside the quantum circuit can be deferred to the end of the computation.

\paragraph{Quantum Oracles.}
Many quantum algorithms require an oracle access. The difference they make with classical algorithms with this respect is that classical oracles (\emph{e.g.} cryptographic oracles such as a cipher with unknown key) are queried ``classically'', with a single value, while quantum oracles are unitary operators. We consider oracle calls of the type:
\[
\Qcircuit @C=1.0em @R=.7em {
\lstick{\ket{x}} & \qw & \multigate{1}{O_f} & \qw & \rstick{\ket{x}} \\
\lstick{\ket{y}} & \qw & \ghost{O_f}         & \qw & \rstick{\ket{y \oplus f(x)}}
}
\]
which XOR their output value to an output register (ensuring reversibility). If we consider that $\ket{y}$ starts in the state $\ket{0}$, then $f(x)$ is simply written here. If the function $f$ can be accessed through $O_f$, we say it has superposition oracle access.

\paragraph{Quantum RAM.}
Additionally to the use of ``plain'' quantum circuits with universal quantum computation, many algorithms require quantum random-access, or being able to access at runtime a \emph{superposition} of memory cells. This is a strong requirement,
since this requires an extensive quantum hardware (the qRAM) and a huge architecture that is harder to build than a quantum circuit with a limited number of qubits.
Shor's algorithm, Simon's algorithm, Grover's algorithm do not require qRAM, if their oracle calls do not either, contrary to, \emph{e.g.}, the algorithm for quantum collision search of~\cite{DBLP:conf/latin/BrassardHT98}, whose optimal speedup can be realized only by using massive qRAM. 

Our algorithm has no such requirement, which puts it on the same level of practicality as Grover's algorithm for attacking symmetric primitives.

\subsection{Simon's Algorithm}

Simon's algorithm~\cite{DBLP:conf/focs/Simon94} gives an exponential speedup on the following problem.

\begin{problem}[Simon's problem]
Suppose given access to a function $f~: \zo^n \rightarrow \zo^n$ that is either injective, or such that there exists $s \in \zo^n$ with:
$$ \forall x, f(x) = f(y) \iff y = x \text{ or } y = x \oplus s, $$
then find $\alpha$.
\end{problem}

In other words, the function $f$ has a hidden Boolean period. It is also easy to extend this algorithm to a hidden Boolean shift, when we want to decide whether two functions $f$ and $g$ are such that $g(x) = f(x \oplus s)$ for all $x$. In practice, $f$ can fall in any set $X$ provided that it can be represented efficiently, but in our examples, we will consider functions producing bit strings.

Solving this problem with classical oracle access to $f$ requires $\Omega \left(2^{n/2} \right)$ queries, as we need to find a collision of $f$ (or none, if there is no hidden period). Simon~\cite{DBLP:conf/focs/Simon94} gives an algorithm which only requires $\bigO{n}$ superposition queries. We fix $c \geq 1$ a small constant to ensure a good success probability and repeat $c n$ times Algorithm~\ref{algorithm:simonproc}.

\begin{algorithm}[h!]
\begin{algorithmic}[1]
\State Start in the all-zero state $\ket{0} \ket{0}$ where the first register contains $n$ qubits and the second represents elements of $X$.
\State Apply Hadamard gates to obtain:
$$ \sum_{x \in \zo^n} \ket{x} \ket{0} $$
\State Query $O_f$ to obtain:
$$ \sum_{x \in \zo^n} \ket{x} \ket{f(x)} = \sum_{a \in X} \left( \sum_{x \in \zo^n | f(x) = a} \ket{x} \right) \ket{a}  $$
\State Measure $a$ (alternatively, we can defer this measurement), get a random value $a \in X$ and:
$$ \sum_{x \in \zo^n | f(x) = a} \ket{x} $$
\State Apply Hadamard gates:
$$ \sum_{y \in \zo^n} \left( \sum_{x \in \zo^n | f(x) = a} (-1)^{x \cdot y} \right) \ket{y} $$
\State Now measure the $y$ register. There are two cases.
\begin{itemize}
\item Either $f$ hides no period $s$, in which case we get a random $y$.
\item Either $f$ hides a period $s$, in which case the amplitude of $\ket{y}$ is: 
$$\sum_{x \in \zo^n | f(x) = a} (-1)^{x \cdot y} = (-1)^{x_0 \cdot y} + (-1)^{(x_0 \oplus s) \cdot y}$$
which is zero if $y \cdot s = 1$ and non-zero otherwise.
\begin{itemize}
\item In that case, measuring gives a random $y$ such that $y \cdot s = 0$.
\end{itemize}
\end{itemize}
\end{algorithmic}

\caption{Quantum subroutine of Simon's algorithm.}
\label{algorithm:simonproc}
\end{algorithm}

We obtain either:
\begin{itemize}
\item a list of $cn$ random values of $y$;
\item a list of $cn$ random values of $y$ in the hyperplane $y \cdot s = 0$.
\end{itemize}
It becomes now easy to test whether $s$ exists or not. If it doesn't, the system of equations obtained has full rank. If it does exist, we can find it by solving the system. 
Judging whether there exists such an $s$ and actually finding it (if it exists) can be done in time $\bigO{n^3}$ by Gaussian elimination.

\paragraph{Simon's Algorithm in Cryptography.}
This algorithm has been used in many attacks on modes of operation and constructions where recovering a secret requires to find a hidden shift between two functions having bit-string inputs. Generally, the functions to which Simon's algorithm is applied are not injective, and random collisions can occur. But a quick analysis (as done \emph{e.g.} in~\cite{DBLP:conf/crypto/KaplanLLN16}) shows that even in this case, a mild increase of the constant $c$ will increase the success probability to a sufficient level.
To be precise, the following proposition holds.
\begin{proposition}[Theorem 2 in \cite{DBLP:conf/crypto/KaplanLLN16}]\label{prop:SimonGen}
Suppose that $f : \{0,1\}^n \rightarrow X$ has a period $s \neq 0^n$, i.e., $f(x\oplus s) = f(x)$ for all $x \in \{0,1\}^n$, and satisfies
\begin{equation}
\max_{t \neq \{s,0^n\}} \Pr_x\left[ f(x \oplus t) = f(x) \right] \leq \frac{1}{2}.
\end{equation}
When we apply Simon's algorithm to $f$, it returns $s$ with a probability at least $1 - 2^n \cdot (3/4)^{cn}$.
\end{proposition}

\subsection{Grover's Algorithm}
Grover's algorithm~\cite{DBLP:conf/stoc/Grover96} allows a quadratic speedup on classical exhaustive search. Precisely, it solves the following problem:

\begin{problem}[Grover's problem]
Consider a set $X$ (the ``search space'') whose elements are represented on $\lceil \log_2(|X|) \rceil$ qubits, such that the uniform superposition $\sum_{x \in X} \ket{x}$ is computable in $\bigO{1}$ time. Given oracle access to a function $f~: X \rightarrow \zo$ (the ``test''), find $x \in X $ such that $f(x) = 1$.
\end{problem}

Classically, if there are $2^t$ preimages of $1$, we expect one to be found in time (and oracle accesses to $f$) $\bigO{|X|/2^t}$. Quantumly, Grover's algorithm finds one in time (and oracle accesses to $O_f$) $\bigOt{ \sqrt{|X| / 2^t}}$. In particular, if there is one preimage of 1, the running time is $\bigOt{ \sqrt{|X|} }$. If the superposition oracle for $f$ uses $a$ ancilla qubits, then Grover's algorithm requires $a + \lceil \log_2(|X|) \rceil$ qubits only.

Grover's algorithm works first by producing the superposition $\sum_{x \in |X|} \ket{x}$. It applies $\bigOt{ \sqrt{|X| / 2^t}}$ times an operator which, by querying $O_f$ ``moves'' some amplitude towards the preimages of $1$.

%% file: algorithm.tex
In this section, we introduce a problem that can be seen as a general combination of Simon's and Grover's problems, and that will be solved by an according combination of algorithmic ideas.
The problem has many cryptographic applications, and it will be at the core of our improved Q2 and Q1 time-memory-data tradeoffs.

\begin{problem}[Asymmetric Search of a Period]\label{problem:search-shift}
Let $F~: \zo^m \times \zo^n \rightarrow \zo^\ell$ and $g~:\zo^n \rightarrow \zo^\ell$ be two functions. We consider $F$ as a family of functions indexed by $\zo^m$ and write $F(i,\cdot) = f_i(\cdot)$. 
Assume that we are given quantum oracle access to $F$, and classical or quantum oracle access to $g$.
(In the Q1 setting, $g$ will be a classical oracle. In the Q2 setting, $g$ will be a quantum oracle.)

Assume that there exists exactly one $i \in \zo^m$ such that $f_{i} \oplus g$ has a hidden period, \emph{i.e.}: $\forall x \in \zo^n, f_{i_0}(x) \oplus g(x) = f_{i_0}(x \oplus s) \oplus g(x \oplus s)$ for some $s$. Furthermore, assume that:
\begin{equation}\label{eq:simoncondition}
\max_{ \substack{i \in \{0,1\}^m \setminus \{i_0\} \\ t \in \{0,1\}^{n} \setminus \{0^{n}\}}} \Pr_{ x \gets \{0,1\}^{n}}\left[ (f_i \oplus g)(x \oplus t)= (f_i \oplus g)(x) \right] \leq \frac{1}{2}
\end{equation}
Then find $i_0$ and $s$.
\end{problem} 

In our cryptographic applications, $g$ will be a keyed function such that adversaries have to make online queries to evaluate it, while $F$ will be a function such that adversaries can evaluate it offline.
For example, the problem of recovering keys of the FX construction $\FX_{k,k_{in},k_{out}}(x) = E_k(x \oplus k_{in}) \oplus k_{out}$ can be regarded as a simple cryptographic instantiation of Problem~\ref{problem:search-shift}:
Set $g(x) := \FX_{k,k_{in},k_{out}}(x)$ and $F(i,x) := E_{i}(x)$.
Then, roughly speaking, the function $f_i \oplus g$ has a period $k_{in}$ if $k=i$, whereas it does not have any period if $i \neq k$ and Condition \eqref{eq:simoncondition} holds. 
Thus we can know whether $i=k$ by checking whether $f_i \oplus g$ has a period.

\paragraph{Justification of Condition~\eqref{eq:simoncondition}.}

We added Condition \eqref{eq:simoncondition} in Problem 3 because the problem would be much harder to solve if we do not suppose any condition on $f_i$. Such assumptions are standard in the litterature of quantum attacks using Simon's algorithm (see for example 
\cite[Sections 2.2 and 4]{DBLP:conf/crypto/KaplanLLN16} or \cite[Section 3]{journals/iacr/Bonnetain17}). This is reasonable for cryptographic applications, as a block cipher is expected to behave like a random permutation, which makes the functions we construct
in our applications behave like random functions. This assumption is made in \cite{DBLP:conf/crypto/KaplanLLN16,DBLP:conf/asiacrypt/Leander017}, and such functions satisfy Condition \eqref{eq:simoncondition} with an overwhelming probability.
Moreover, as remarked in \cite{DBLP:conf/crypto/KaplanLLN16}, a cryptographic construction that fails to satisfy Condition \eqref{eq:simoncondition} would exhibit some poor differential properties which 
could be used for cryptanalysis.

\subsection{Existing Techniques to Solve the Problem}\label{sec:ExistingAlg}
Here we explain existing algorithms to solve Problem~\ref{problem:search-shift} in both the Q1 model and the Q2 model, with the algorithms to recover keys of the FX construction as an example.
Note that we consider the situation in which exponentially many qubits are \emph{not} available.
\paragraph{The model Q1.}
In the Q1 model, when we are allowed to make only classical queries to $g := \FX_{k,k_{in},k_{out}}$, there exists a Q1 algorithm to attack the FX construction that uses a kind of meet-in-the-middle technique~\cite{DBLP:conf/ctrsa/HosoyamadaS18}.
However, it does not make use of the exponential speed-up of Simon's algorithm, and its time complexity and query complexity is $\bigO{2^{3(n+m)/7}}$ (for $m \leq 4n/3$).

\paragraph{The model Q2.}
Problem~\ref{problem:search-shift} can be solved with $ \bigO{n 2^{m/2}} $ superposition queries to $F(i,x)=E_i(x)$ and $g(x)=\FX_{k,k_{in},k_{out}}(x)$, and in time $ \bigO{n^3 2^{m/2}} $, using the Grover-meet-Simon algorithm of~\cite{DBLP:conf/asiacrypt/Leander017}. Indeed, we make a Grover search on index $i \in \zo^m$. When testing whether a guess $i$ for the key $k$ is correct, we perform $\bigO{n}$ queries to $F$ and $\bigO{n}$ queries to $g$, to check whether $f_i \oplus g$ has a hidden period, hence whether the guess $i$ is correct. Moreover, since superposition access to $F$ and $g$ is allowed, we can test $i$ in superposition as well.

\subsection{An Algorithm for Asymmetric Search of a Shift}

Here we describe our new algorithms to solve Problem~\ref{problem:search-shift}.
We begin with explaining two observations on the Grover-meets-Simon algorithm in the Q2 model described in Section~\ref{sec:ExistingAlg}, and how to improve it.
Then we describe how to use the idea to make a good algorithm to solve Problem~\ref{problem:search-shift} in the Q1 model.

\subsubsection{Two observations.}
Our first observation is that, when doing the Grover search over $i$ for Problem~\ref{problem:search-shift}, each time a new $i$ is tested, a new function $f_i$ is queried. But, in contrast, the function $g$ is always the same. We would like to take this asymmetry into account, namely, to make less queries to $g$ since it does not change. This in turn has many advantages: queries to $g$ can become more costly than queries to $f_i$. 

Our second observation is that, for each $i \in I$, once we have a superposition $\ket{\psi_g} = \bigotimes^{cn} \left( \sum_{x \in \zo^n}   \ket{x} \ket{g(x)} \right)$ and given a quantum oracle access to $f_i$, we can obtain the information if $f_i \oplus g$ has a period or not without making queries to $g$.

From $\ket{\psi_g}$, we can make the state $\ket{\psi_{f_i \oplus g}} = \bigotimes^{cn} \left( \sum_{x \in \zo^n} \ket{x} \ket{f_i(x) \oplus g(x)} \right)$ by making queries to $f_i$.
By applying usual Simon's procedures on $\ket{\psi_{f_i \oplus g}}$, we can judge if $f_i \oplus g$ has a period.
Moreover, by appropriately performing uncomputations, we can recover $\ket{\psi_g}$ (with some errors) and reuse it in other procedures.

With these observations in mind, below we give an intuitive description of our algorithm $\AlgQtwo$ to solve Problem~\ref{problem:search-shift} in the Q2 model (we name our algorithm $\AlgQtwo$ because it will be applied to make Q2 attacks with \emph{polynomially} many online queries in later sections).
The main ideas of $\AlgQtwo$ are separating an online phase and offline computations, and iteratively reusing the quantum data $\ket{\psi_g}$ obtained by the online phase.

\subsubsection{Algorithm \AlgQtwo (informal).}
\begin{enumerate}
\item Online phase: Make $cn$ quantum queries to $g$ to prepare $\ket{\psi_g}$.
\item Offline computations: Run the Grover search over $i \in \{0,1\}^m$.
For each fixed $i$, run a testing procedure ${\sf test}$ such that: (a) ${\sf test}$ checks if $i$ is a good element (i.e., $f_i \oplus g$ has a period) by using $\ket{\psi_g}$ and making queries to $f_i$, and (b)
after checking if $i$ is good, appropriately performs uncomputations to recover the quantum data $\ket{\psi_g}$.
\end{enumerate}
A formal description of $\AlgQtwo$ is given in Algorithm~\ref{algorithm:asymmetric}.
We fix a constant $c \geq 1$, to be set later depending on the probability of error wanted.

\begin{algorithm}[h!]
\begin{algorithmic}[1]
\State Start in the all-zero state.
\State Using $cn$ queries to $g$, create the state:
$$ \ket{\psi_g} = \bigotimes^{cn} \left( \sum_{x \in \zo^n}   \ket{x} \ket{g(x)} \right) $$

The circuit now contains $\ket{\psi_g}$, the ``g-database'', and additional registers on which we can perform Grover search. Notice that $\ket{\psi_g}$ contains $cn$ independent (and disentangled) registers.
\State Create the uniform superposition over indices $i \in \zo^m$:
$$ \ket{\psi_g} \otimes \sum_{i \in \zo^m} \ket{i} $$

\State Apply Grover iterations. The testing oracle is a unitary operator $\mathsf{test}$ that takes in input a register for $\ket{i}$ and the ``g-database'', and tests in superposition whether $f_i \oplus g$ has a hidden period.
If this is the case, it returns $\ket{b \oplus 1}$ on input $\ket{b}$. Otherwise it returns $\ket{b}$.
(Algorithm~\ref{algorithm:test} gives the details for $\sf test$ in the case that $i$ is fixed.)

\vspace*{-0.5cm}
\begin{center}
\[
\Qcircuit @C=1.0em @R=.7em {
\lstick{\ket{\psi_g}} & \qw & \multigate{2}{\mathsf{test}} & \qw & \rstick{\ket{\psi_g}} \\
\lstick{\ket{i}} & \qw & \ghost{\mathsf{test}}         & \qw & \rstick{\ket{i}} \\
\lstick{\ket{b}} & \qw & \ghost{\mathsf{test}}  & \qw & \rstick{\ket{ b \text{ or } b \oplus 1}}
}
\]
\end{center}
\vspace*{0.1cm}

The most important feature of $\mathsf{test}$ is that it does not change the g-database (up to some errors). The registers holding $\ket{\psi_g}$ are disentangled before and after the application of $\mathsf{test}$.
\State After $\bigO{2^{m/2}}$ Grover iterations, measure the index $i$.
\State If the hidden shift is also wanted, apply a single instance of Simon's algorithm (or re-use the database and perform a slightly extended computation of $\mathsf{test}$ to retrieve the result).
\end{algorithmic}
\caption{\AlgQtwo.}

\label{algorithm:asymmetric}
\end{algorithm}

We show how to implement the testing procedure $\mathsf{test}$ in Algorithm~\ref{algorithm:test} without any new query to $g$, using only exactly $2cn$ superposition queries to $F$. To write this procedure clearly, we consider a single function $f$ in input, but remark that it works as well if $f$ is a superposition of $f_i$ (as will be the case when $\mathsf{test}$ is called as the oracle of a Grover search).

\begin{algorithm}[h!]
\begin{algorithmic}[1]
\State We start with the g-database:
$$ \ket{\psi_g} = \bigotimes^{cn} \left( \sum_{x \in \zo^n}  \ket{x} \ket{g(x)} \right) $$

\State Using $cn$ superposition queries to $f$, build the state:
$$ \ket{\psi_{f \oplus g}} = \bigotimes^{cn} \left( \sum_{x \in \zo^n} \ket{x} \ket{g(x) \oplus f(x)} \right) $$

We will now perform, in a reversible way, the exact computations of Simon's algorithm to find if $g \oplus f$ has a hidden period or not (in which case $f$ and $g$ have a hidden shift).

\State Apply $\left(H^{\otimes n} \otimes I_m \right)^{cn} \otimes I_1$ to $\ket{\psi_{f \oplus g}} \otimes \ket{b}$, to obtain
\begin{align}
&\left( \sum_{u_1,x_1 \in \{0,1\}^{n}} (-1)^{u_1 \cdot x_1} \ket{u_1} \ket{(f \oplus g) (x_1)} \right) \otimes \cdots \nonumber \\
&\hspace{50pt} \cdots \otimes
\left( \sum_{u_{cn},x_{cn} \in \{0,1\}^{n}} (-1)^{u_{cn} \cdot x_{cn}} \ket{u_{cn}} \ket{(f \oplus g) (x_{cn})} \right) \otimes \ket{b}.
\end{align}
\State Compute $d := \dim({\rm Span}(u_1,\dots,u_{cn}))$, set $r:=0$ if $d = n$ and $r:=1$ if $d < n$, and add $r$ to $b$.
Then uncompute $d$ and $r$, and obtain
\begin{align}\label{eq:Alg3Step4}
&\sum_{\substack{u_1, \dots, u_{cn} \\ x_1,\dots,x_{cn}}} (-1)^{u_1 \cdot x_1} \ket{u_1} \ket{(f \oplus g) (x_{1})} \otimes \cdots \nonumber \\
&\hspace{50pt} \cdots \otimes (-1)^{u_{cn} \cdot x_{cn}} \ket{u_{cn}} \ket{(f \oplus g) (x_{cn})}  \otimes \ket{b \oplus r}.
\end{align}
Note that $r$ in \eqref{eq:Alg3Step4} depends on $u_1,\dots,u_{cn}$ and now the last register may be entangled with the registers of $u_1,\dots,u_{cn}$.
\State Uncompute $\left(H^{\otimes n} \otimes I_m \right)^{cn} \otimes I_1$.
\State Using $cn$ new superposition queries to $f$, revert $\ket{\psi_{f \oplus g}}$ to $\ket{\psi_g}$.

There are two cases:
\begin{itemize}
\item If $f \oplus g$ has a hidden period, then $r = 1$ always holds. Hence, in the output register, we always write  $1$.
\item If $f \oplus g$ does not have a hidden period, then with high probability, $r = 0$. Hence, in the output register, we write $0$.
\end{itemize}

\end{algorithmic}

\caption{The procedure {\sf test} that checks if a function $f \oplus g$ has a period against the g-database, without any new query to $g$.}
\label{algorithm:test}
\end{algorithm}

In practice, Algorithm~\ref{algorithm:test} works up to some error (see Remark~\ref{remark:errorbound}), which is amplified at each iteration of Algorithm~\ref{algorithm:asymmetric}.
The complexity and success probability (including the errors) of $\AlgQtwo$ can be analyzed as below.
\begin{proposition}\label{prop:AlgQ2}
Suppose that $m$ is in $\bigO{n}$.
Let $c$ be a sufficiently large constant. \footnote{See Proposition~\ref{prop:cval} for a concrete estimate.} Consider the setting of Problem~\ref{problem:search-shift}: let $i_0 \in \{0,1\}^m$ be the good element such that $g \oplus f_{i_0}$ is periodic and assume that~\eqref{eq:simoncondition} holds.
Then $\AlgQtwo$ finds $i_0$ with a probability in $\Theta(1)$ by making $\bigO{n}$ quantum queries to $g$ and $\bigO{n2^{m/2}}$ quantum queries to $F$.
\footnote{
In later applications, $F$ will be instantiated with unkeyed primitives, and quantum queries to $F$ are emulated with offline computations of primitives such as block ciphers.
}
The offline computation (the procedures excluding the ones to prepare the state $\ket{\psi_g}$) of $\AlgQtwo$ is done in time $\bigO{(n^3+nT_F)2^{m/2}}$, where $T_F$ is the time required to evaluate $F$ once.
\end{proposition}

See Section~A in the \acversion{full version of the paper~\cite{DBLP:journals/iacr/BonnetainHNSS19}}\fullversion{Appendix} for a proof.

\begin{remark}\label{remark:errorbound}
Intuitively, the error produced in each iteration of Algorithm~\ref{algorithm:test} is bounded by the maximum, on $i$, of:
$
p^{(i)}:= \Pr \left[ \dim({\rm Span}(u_1,\dots,u_{cn})) < n\right],
$
when $u_1,\dots,u_{cn}$ are produced with Simon's algorithm, \emph{i.e.} the probability that Simon's algorithm returns the incorrect answer ``$f_i \oplus g$ is periodic'' even though $f_i \oplus g$ is not periodic.
From condition~\eqref{eq:simoncondition}, we can show that $p^{(i)} \leq 2^{(n+1)/2} ((1 + \frac{1}{2}) / 2)^{cn/2}$ holds (see Lemma~1 in the \acversion{full version of the paper~\cite{DBLP:journals/iacr/BonnetainHNSS19}}\fullversion{Appendix}).
\end{remark}

\begin{remark}
$\AlgQtwo$ finds the index $i$ such that $f_i \oplus g$ has a period,  but does not return the actual period of $f_i \oplus g$.
However, we can find the actual period of $f_i \oplus g$ (after finding $i$ with $\AlgQtwo$) by applying Simon's algorithm to $f_i \oplus g$.
\end{remark}

\paragraph{Summary.}
With $\AlgQtwo$, we realize an ``asymmetric'' variant of Simon's algorithm, in which we store a ``compressed'' database for a single function $g$, which is not modified (up to some errors) while we test whether another function $f$ has a hidden shift with $g$, or not. An immediate application of this algorithm will be to 
achieve an exponential improvement of the query complexity of some Q2 attacks on symmetric schemes. Indeed, in the context where Simon's algorithm and Grover's algorithm are combined, it may be possible to perform the queries to the secret-key cryptographic oracle only once, and so, to lower the query complexity to $\bigO{n}$. 

\subsection{Asymmetric Search with Q1 Queries}
In $\AlgQtwo$, (online) queries to $g$ and (offline) queries to $F$ are separated, and only $cn$ superposition queries to $g$ are made.
Hence another tradeoff is at our reach, which was not possible when $g$ was queried in each Grover iteration: removing superposition queries to $g$ completely.

\begin{algorithm}[h!]
\begin{algorithmic}[1]
\Statex \textbf{Input: } Classical query access to $g$
\Statex \textbf{Output: } The g-database:
$$ \ket{\psi_g} = \bigotimes^{cn} \left( \sum_{x \in \zo^n} \ket{x} \ket{g(x)} \right) $$

\State Start with the all-zero state 
$$ \bigotimes^{cn} \ket{0} \ket{0} $$

\State Apply Hadamard gates:
$$ \bigotimes^{cn} \sum_{x \in \zo^n}  \ket{x} \ket{0} $$

\State For each $y \in \zo^n$, query (classically) $g(y)$, then apply a unitary which writes $g(y)$ in the second register if the first contains the value $y$.
\end{algorithmic}

\caption{Producing the g-database $\ket{\psi_g}$.}
\label{algorithm:database}
\end{algorithm}

This requires now to query \emph{the whole codebook} for $g$ to prepare the quantum state $\ket{\psi_g}$.
Once $\ket{\psi_g}$ is built, the second offline phase runs in exactly the same way.
Building $\ket{\psi_g}$ costs roughly $2^n$ time (and classical queries), while going through the search space for $f$ takes $2^{m/2}$ iterations (and quantum queries to $F$).
We call our new algorithm in the Q1 model $\AlgQone$ because it will be applied to make Q1 attacks with \emph{exponentially} many online queries in later sections.
The optimal point arrives when $m = 2n$.

Below we give an intuitive description of our algorithm $\AlgQone$ to solve Problem~\ref{problem:search-shift} in the Q1 model.
As described above, the difference between $\AlgQone$ and $\AlgQtwo$ is the online phase to prepare $\ket{\psi_g}$.
\subsubsection{Algorithm \AlgQone (informal).}
\begin{enumerate}
\item Online phase: Make $2^n$ classical queries to $g$ and prepare the state $\ket{\psi_g}$.
\item Offline computations: Run the Grover search over $i \in \{0,1\}^m$.
For each fixed $i$, run a testing procedure ${\sf test}$ such that: (a) ${\sf test}$ checks if $i$ is a good element (i.e., $f_i \oplus g$ has a period) by using $\ket{\psi_g}$ and making queries $f_i$, and (b)
after checking if $i$ is good, appropriately perform uncomputations to recover the quantum data $\ket{\psi_g}$.
\end{enumerate}
A formal description of $\AlgQone$ is the same as that of $\AlgQtwo$ (Algorithm~\ref{algorithm:asymmetric}) except that we make $2^n$ classical queries to $g$ to prepare the quantum state $\ket{\psi_g}$.
See Algorithm~\ref{algorithm:database} for formal description of the online phase.
The complexity and success probability (including the errors) of $\AlgQone$ can be analyzed as below.
\begin{proposition}\label{prop:AlgQ1}
Suppose that $m$ is in $\bigO{n}$.
Let $c$ be a sufficiently large constant. \footnote{See Proposition~\ref{prop:cval} for a concrete estimate.} Consider the setting of Problem~\ref{problem:search-shift}: let $i_0 \in \{0,1\}^m$ be the good element such that $g \oplus f_{i_0}$ is periodic and assume that~\eqref{eq:simoncondition} holds.
Then $\AlgQone$ finds $i_0$ with a probability in $\Theta(1)$ by making $\bigO{2^n}$ classical queries to $g$ and $\bigO{n2^{m/2}}$ quantum queries to $F$.
\footnote{
Again, in later applications, $F$ will be instantiated with unkeyed primitives, and quantum queries to $F$ are emulated with offline computations of primitives such as block ciphers.
}
The offline computation (the procedures excluding the ones to prepare the state $\ket{\psi_g}$) of $\AlgQone$ is done in time $\bigO{(n^3+nT_F)2^{m/2}}$, where $T_F$ is the time required to evaluate $F$ once.
\end{proposition}
A proof is given in Section~A in the \acversion{full version of the paper~\cite{DBLP:journals/iacr/BonnetainHNSS19}}\fullversion{Appendix}.

\subsubsection{Finding actual periods.}
The above algorithm $\AlgQone$ returns the index $i_0$ such that $f_{i_0} \oplus g$ has a period, but does not return the actual period.
Therefore, if we want to find the actual period of $f_{i_0} \oplus g$ after finding $i_0$, we have to apply Simon's algorithm to $f_{i_0} \oplus g$ again.
Now we can make only classical queries to $g$, though, we can use the same idea with $\AlgQone$ to make an algorithm $\SimQone$ that finds the period of $f_{i_0} \oplus g$.
Again, let $c$ be a positive integer constant.\\

\noindent\textbf{Algorithm $\SimQone$.}
\begin{enumerate}
\item Make $2^n$ classical queries to $g$ and prepare the quantum state $\ket{\psi_g}$.
\item Make $cn$ quantum queries to $f_{i_0}$ to obtain the quantum state $\ket{\psi_{f_{i_0} \oplus g}} = \bigotimes^{cn} \left( \sum_x \ket{x} \ket{f_{i_0}(x) \oplus g (x)} \right) $.
\item Apply $H^{\otimes n}$ to each $\ket{x}$ register to obtain the state 
$$\bigotimes^{cn} \left( \sum_{x,u} (-1)^{x \cdot u}\ket{u} \ket{f_{i_0}(x) \oplus g (x)} \right) \enspace.$$
\item Measure all $\ket{u}$ registers to obtain $cn$ vectors $u_1,\dots,u_{cn}$.
\item Compute the dimension $d$ of the vector space $V$ spanned by $u_1,\dots,u_{cn}$. If $d \neq n-1$, return $\perp$. If $d = n-1$, compute the vector $v \neq 0^n \in \{0,1\}^n$ that is orthogonal to $V$.
\end{enumerate} 
Obviously the probability that the above algorithm $\SimQone$ returns the period of $f_{i_0} \oplus g$ is the same as the probability that the original Simon's algorithm returns the period, under the condition that $cn$ quantum queries can be made to the function $f_{i_0} \oplus g$.
Thus, from Proposition~\ref{prop:SimonGen}, the following proposition holds.

\begin{proposition}\label{prop:SimQ1}
Suppose that $f_{i_0} \oplus g$ has a period $s \neq 0^n$ and satisfies
\begin{equation}
\max_{t \neq \{s,0^n\}} \Pr_x\left[ (f_{i_0} \oplus g)(x \oplus t) = (f_{i_0} \oplus g)(x) \right] \leq \frac{1}{2}.
\end{equation}
Then $\SimQone$ returns $s$ with a probability at least $1 - 2^n \cdot (3/4)^{cn}$ by making $\bigO{2^n}$ classical queries to $g$ and $cn$ quantum queries to $f_{i_0}$.
The offline computation of $\SimQone$ (the procedures excluding the ones to prepare the state $\ket{\psi_g}$) runs in time $\bigO{n^3 + nT_f}$, where $T_f$ is the time required to evaluate $f_{i_0}$ once.
\end{proposition}

\begin{proposition}[Concrete cost estimates]\label{prop:cval}
 In practice, for Propositions~\ref{prop:AlgQ2} and~\ref{prop:AlgQ1}, $c\simeq m/\left(n\log_2(4/3)\right)$ is sufficient.
\end{proposition}
\begin{proof}
 We need to have $4\lfloor \pi / \left(4\arcsin\left(2^{-m/2}\right)\right) \rfloor 2^{(n+1)/2}(3/4)^{cn/2}< 1/2$.
 
 In practice, $\arcsin(x) \simeq x$ and the rounding has a negligible impact. Hence, we need that
 $m/2+(n+1)/2+\log_2(\pi) + \log_2(3/4)cn/2 < -1$.
 
 This reduces to $c > \log_2(4/3)^{-1}\left(m+3+2\log_2(\pi)\right)/n \simeq m/\left(n\log_2(4/3)n\right)$.
\end{proof}

\begin{remark}
  If $m = n$, this means $ c\simeq 2.5$, and if $m = 2n$, $c \simeq 5$.
\end{remark}

%% file: q2applications.tex

This section shows that our new algorithm $\AlgQtwo$ can be used to construct Q2 attacks on various symmetric schemes.
By using $\AlgQtwo$ we can exponentially reduce the number of quantum queries to the keyed oracle compared to previous Q2 attacks, with the same time cost.

In each application, we consider that one evaluation of each primitive (e.g., a block cipher) can be done in time $\bigO{1}$, for simplicity. For our practical estimates, we use the cost of the primitive as our unit, and consider that it is greater than the cost of solving the linear equation system.
In addition, we assume that key lengths of $n$-bit block ciphers are in $\bigO{n}$, which is the case for usual block ciphers.

\subsection{An Attack on the FX construction}
Here we show a Q2 attack on the FX construction.
As described in Section~\ref{sec:algorithm}, the FX construction builds an $n$-bit block cipher $\mathrm{FX}_{k,k_{in},k_{out}}$ with $(2n+m)$-bit keys ($k_{in},k_{out} \in \{0,1\}^n$ and $k \in \{0,1\}^m$) from another $n$-bit block cipher $E_{k}$ with $m$-bit keys as
\begin{equation}
\mathrm{FX}_{k,k_{in},k_{out}} (x) := E_k(x \oplus k_{in}) \oplus k_{out}.
\end{equation}

This construction is used to obtain a block cipher with long ($(2n+m)$-bit) keys from another block cipher with short ($m$-bit) keys.
Roughly speaking, in the classical setting, the construction is proven to be secure up to $\bigO{2^{(n+m)/2}}$ queries and computations if the underlying block cipher is secure~\cite{DBLP:conf/crypto/KilianR96}.

Concrete block ciphers such as DESX, proposed by Rivest in 1984 and analyzed in~\cite{DBLP:conf/crypto/KilianR96}, PRINCE~\cite{DBLP:conf/asiacrypt/BorghoffCGKKKLNPRRTY12}, and PRIDE~\cite{DBLP:conf/crypto/AlbrechtDKLPY14} are designed based on the FX construction.
To estimate security of these block ciphers against quantum computers, it is important to study quantum attacks on the FX construction.

As briefly explained in Section~\ref{sec:algorithm}, the previous Q2 attack by Leander and May~\cite{DBLP:conf/asiacrypt/Leander017} breaks the FX construction by making $\bigO{ n2^{m/2}}$ quantum queries, and its time complexity is $\bigO{n^32^{m/2}}$.

\subsubsection{Application of our algorithm $\AlgQtwo$.}
Below we show that, by applying our algorithm $\AlgQtwo$, we can recover keys of the FX construction with only $\bigO{n}$ quantum queries. 
Time complexity of our attack remains $\bigO{n^32^{m/2}}$, which is the same as Leander and May's.

\paragraph{Attack idea.}
As explained in Section~\ref{sec:algorithm}, the problem of recovering the keys $k$ and $k_{in}$ of the FX construction $F_{k,k_{in},k_{out}}$ can be reduced to Problem~\ref{problem:search-shift}:
Define $F : \{0,1\}^m \times \{0,1\}^n \rightarrow \{0,1\}^n$ and $g : \{0,1\}^n \rightarrow \{0,1\}^n$ by 
\[
 \begin{matrix} F(i,x) &=& E_i(x) \oplus E_i(x\oplus 1)\\
 g(x) &=& \FX_{k,k_{in}, k_{out}}(x)\oplus \FX_{k,k_{in}, k_{out}}(x\oplus1).\end{matrix}
\]
Then 
\begin{equation}
f_k(x) \oplus g(x) = f_k(x \oplus k_{in}) \oplus g(x \oplus k_{in})
\end{equation}
holds, i.e., $f_k \oplus g(x)$ has a period $k_{in}$ (note that $f_k(x) = F(k,x)$).
If $E$ is a secure block cipher and $E_i$ is a random permutation for each $i$, intuitively, $f_{i} \oplus g$ does not have any period for $i \neq k$.
Thus the problem of recovering $k$ and $k_{in}$ is reduced to Problem~\ref{problem:search-shift} and we can apply our algorithm $\AlgQtwo$.
Formally, the attack procedure is as follows.
\paragraph{Attack description.}
\begin{enumerate}
\item Run $\AlgQtwo$ for the above $F$ and $g$ to recover $k$.
\item Apply Simon's algorithm to $f_k \oplus g$ to recover $k_{in}$.
\item Compute $k_{out} = E_k(0^n) \oplus \FX_{k,k_{in}, k_{out}}(0^n)$.
\end{enumerate}

Next we give a complexity analysis of the above attack.
\paragraph{Analysis.}
We assume that $m = \bigO{n}$, which is the case for usual block ciphers.
If $E$ is a secure block cipher and $E_i$ is a random permutation for each $i \in \{0,1\}^m$, we can assume that $f_k \oplus g = E_k \oplus E_k(\cdot \oplus 1) \oplus \FX_{k,k_{in}, k_{out}} \oplus \FX_{k,k_{in}, k_{out}}(\cdot\oplus1) $ is far from periodic for all $i \neq k$, and that assumption~\eqref{eq:simoncondition} in Problem~\ref{problem:search-shift} holds.

Hence, by Proposition~\ref{prop:AlgQ2}, $\AlgQtwo$ recovers $k$ with a high probability by making $\bigO{n}$ quantum queries to $g$ and $\bigO{n2^{m/2}}$ quantum queries to $F$, which implies that $k$ is recovered only with $\bigO{n}$ quantum queries made to $\FX_{k,k_{in},k_{out}}$, and in time $\bigO{n^3 2^{m/2}}$.
(Note that one evaluation of $g$ (resp., $F$) can be done by $\bigO{1}$ evaluations of $\FX_{k,k_{in},k_{out}}$ (resp., $E$).)

From Proposition~\ref{prop:SimonGen}, the second step can be done with $\bigO{n}$ quantum queries in time $\bigO{n^3}$. It is obvious that the third step can be done efficiently.

In summary, our attack recovers the keys of the FX construction with a high probability by making $\bigO{n}$ quantum queries to the (keyed) online oracle, and it runs in time $\bigO{n^32^{m/2}}$.

\subsubsection{Applications to DESX, PRINCE and PRIDE.}
DESX~\cite{DBLP:conf/crypto/KilianR96} has a 64-bit state, two 64-bit whitening key and one
56-bit inner key. From Propositions~\ref{prop:AlgQ2} and~\ref{prop:cval}, we can estimate that our attack needs roughly $135$ quantum queries and $2^{29}$ quantum computations of the cipher circuit. 

PRINCE~\cite{DBLP:conf/asiacrypt/BorghoffCGKKKLNPRRTY12}, and PRIDE~\cite{DBLP:conf/crypto/AlbrechtDKLPY14} are two ciphers using the FX construction with a 64-bit state, a 64-bit inner key and two 64-bit whitening keys. Hence, from Propositions~\ref{prop:AlgQ2} and~\ref{prop:cval}, we can estimate that our attack needs roughly $155$ quantum queries and $2^{33}$ quantum computations of the cipher circuit.

%% file: q1applications.tex
This section shows that our new algorithm $\AlgQone$ can be used to construct Q1 attacks on various symmetric schemes, with a tradeoff between online classical queries, denoted below by $D$, and offline quantum computations, denoted below by $T$.

All the algorithms that we consider run with a single processor, but they can use quantum or classical memories, whose amount is respectively denoted by $Q$ (number of qubits) and $M$.
Again, we consider that one evaluation of each primitive (\emph{e.g.} a block cipher) can be done in time $\bigO{1}$, for simplicity, and we assume that key lengths of $n$-bit block ciphers are in $\bigO{n}$.

\subsection{Tradeoffs for the Even-Mansour Construction}\label{sec:EMQ1}
The Even-Mansour construction~\cite{DBLP:journals/joc/EvenM97} is a simple construction to make an $n$-bit block cipher $E_{k_1,k_2}$ from an $n$-bit public permutation $P$ and two $n$-bit keys $k_1,k_2$ (see Fig.~\ref{fig:EM}).
The encryption $E_{k_1,k_2}$ is defined as $E_{k_1,k_2}(x) := P(x \oplus k_1) \oplus k_2$, and the decryption is defined accordingly.

In the classical setting, roughly speaking, the Even-Mansour construction is proven secure up to $\bigO{2^{n/2}}$ online queries and offline computations~\cite{DBLP:journals/joc/EvenM97}.
In fact there exists a classical attack with tradeoff $TD=2^n$, which balances at $T=D=2^{n/2}$~\cite{DBLP:conf/asiacrypt/Daemen91}.

\begin{figure}
\centering
\begin{tikzpicture}
  \draw
  node at (0,0)[l,name=m] {$x$}
  node [xor,circle, name=xor1, right of=m] {}
  node [block,name=f1, right of=xor1] {$P$}
  node [xor,circle, name=xor2, right of=f1] {}
  node [name=k1, above of=xor1] {$k_1$}
  node [name=k2, above of=xor2] {$k_2$}
  node [name=out, right of=xor2, right] {$E_{k_1,k_2}(x)$};
  \draw[->] (m) -- (f1) -- (out);
  \draw[->] (k1) -- (xor1);
  \draw[->] (k2) -- (xor2);
 \end{tikzpicture}
\caption{The Even-Mansour construction.}
\label{fig:EM}
\end{figure}
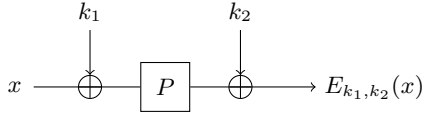

\subsubsection{Previous Q1 attacks on the Even-Mansour construction.}
Kuwakado and Morii gave a Q1 attack that recovers keys of the Even-Mansour construction with $\bigO{2^{n/3}}$ classical queries and qubits, and $\bigO{2^{n/3}}$ offline quantum computations~\cite{DBLP:conf/isita/KuwakadoM12}.
Their attack is based on a claw-finding algorithm by Brassard \emph{et al.}~\cite{DBLP:conf/latin/BrassardHT98}, and gives the tradeoff $T^2 D=2^n$, with additional  $Q=D$ qubits.
The balanced point $2^{n/3}$ is significantly smaller than the classical balanced point $2^{n/2}$.
However, if we want to recover keys with this attack in time $T \ll 2^{n/2}$, we need an exponential amount of qubits.

\paragraph{Main previous attacks with polynomial qubits}

The best classical attacks allow a trade-off of $D \cdot T= 2^n$ (see~\cite{DBLP:conf/asiacrypt/DinurDKS14} for other trade-offs involving memory). With Grover we could recover the keys with a complexity of $2^{n/2}$ and 2 plaintexts-ciphertext pairs, $(P_1,C_1)$ and $(P_2,C_2)$, by performing an exhaustive search over the value of $k_1$ that would verify $P(P_1\oplus k_1)\oplus P(P_2\oplus k_1)= C_1\oplus C_2$.
In~\cite{DBLP:conf/ctrsa/HosoyamadaS18}, Hosoyamada and Sasaki also gave a tradeoff $D \cdot T^6=2^{3n}$ for $D \leq 2^{3n/7}$ under the condition that only polynomially many qubits are available, by using the multi-target preimage search by Chailloux et al~\cite{DBLP:conf/asiacrypt/ChaillouxNS17}.
$D$ and $T$ are balanced at $D=T=2^{3n/7}$, which is smaller than the classical balanced point $2^{n/2}$.
The attack uses only polynomially many qubits, but requires $M=D^{1/3} = 2^{n/7}$ classical memory. At the balanced point, this still represents an exponentially large storage.
Note that this is the only previous work that recover keys in time $T \ll 2^{n/2}$ with polynomially many qubits.

\begin{table}[htb]
\centering
\begin{tabular}{|c||c|c|c|c|c|c|}
\toprule
Reference & Classical attack & Grover & \cite{DBLP:conf/ctrsa/HosoyamadaS18} & \cite{DBLP:conf/isita/KuwakadoM12}& \bf{[Ours]} \\
\midrule
\begin{tabular}{c}
Tradeoff of \\ $D$ and $T$
\end{tabular}
& $D \cdot T=2^n$ &

\begin{tabular}{c}
\\
$T=2^{n/2}$, \\
$D =\mathrm{constant}$
\end{tabular}
&
\begin{tabular}{c}
\\
$D \cdot T^6=2^{3n}$ \\
$(D \leq 2^{3n/7})$
\end{tabular}
& \begin{tabular}{c} \\ $D=2^{n/3}$, \\ $T=2^{n/3}$ \end{tabular}
& $D \cdot T^2 = 2^n$ \\
\midrule
\begin{tabular}{c}
Num. of qubits
\end{tabular}
& - & $\mathsf{poly}(n)$ & $\mathsf{poly}(n)$ & $2^{n/3}$ & $\mathsf{poly}(n)$ \\
\midrule
\begin{tabular}{c}
Classical memory 
\end{tabular}
& $D$ & ${\mathsf{poly}(n)}$ & $D^{1/3}$ & ${\mathsf{poly}(n)}$ & ${\mathsf{poly}(n)}$ \\
\midrule
\begin{tabular}{c}
Balanced point \\ of $D$ and $T$
\end{tabular} &
$2^{n/2}$ & - & $2^{3n/7}$ & - &  $2^{n/3}$ \\
\bottomrule
\end{tabular}
\medskip
\caption{Tradeoffs for Q1 attacks on the Even-Mansour construction.
In this table we omit to write order notations, and ignore polynomial factors in the first and last rows.}
\label{table:tradeoffEM}
\end{table}

\subsubsection{Application of $\AlgQone$.}
We explain how to use our algorithm $\AlgQone$ to attack the Even-Mansour construction. The tradeoff that we obtain is  $T^2 \cdot D = 2^{n}$, the same as the attack by Kuwakado and Morii above. It balances at $T=D=2^{n/3}$, but we use only $\mathsf{poly}(n)$ qubits and $\mathsf{poly}(n)$ classical memory.
See Table~\ref{table:tradeoffEM} for comparison of attack complexities under the condition that $\mathsf{poly}(n)$ many qubits are available.
\paragraph{Attack idea.}

The core observation of Kuwakado and Morii's polynomial-time attack in the Q2 model~\cite{DBLP:conf/isita/KuwakadoM12} is that the $n$-bit secret key $k_1$ is the period of the function $E_{k_1,k_2}(x) \oplus P(x)$, and thus Simon's algorithm can be applied if quantum queries to $E_{k_1,k_2}$ are allowed. The key to this exponential speed up (compared to the classical attack) is to exploit the algebraic structure of $E_{k_1,k_2}$ (the hidden period of the function) with Simon's algorithm.

On the other hand, the previous Q1 (classical query) attacks described above use only generic multi-target preimage search algorithms that do not exploit any algebraic structures. Hence being able to exploit the algebraic structure in the Q1 model should give us some advantage.

Our algorithm $\AlgQone$ realizes this idea. It first makes classical online queries to emulate the quantum queries required by Simon's algorithm (the g-database above) and then runs a combination of Simon's and Grover's algorithms offline (Grover search is used to find the additional $m$-bit secret information). A naive way to attack in the Q1 model would be to immediately combine Kuwakado and Morii's Q2 attack with $\AlgQone$. However, we would have to query the whole classical codebook to emulate quantum queries, which is too costly (and there is no Grover search step).

Our new attack is as follows:
We divide the $n$-bit key $k_1$ in $k^{(1)}_1$ of $u$ bits and $k^{(2)}_1$ of ${n-u}$ bits and apply Simon's algorithm to recover $k^{(1)}_1$, while we guess $k^{(2)}_1$ by the Grover search (see Fig.~\ref{fig:Q1EM}).
Then, roughly speaking, $\AlgQone$ recovers the key by making $D=2^u$ classical queries and $T=2^{(n-u)/2}$ offline Grover search iterations (note that the offline computation cost for Simon's algorithm is $\mathsf{poly}(n)$ and we ignore polynomial factors here for simplicity), which yields the tradeoff $D \cdot T^2 = 2^n$, only with $\mathsf{poly}(n)$ qubits and $\mathsf{poly}(n)$ classical space.

\usetikzlibrary{arrows}
\renewcommand{\midsection}{\tikz \draw[-] (-0.1,-0.1) -- (0.1,0.1);}
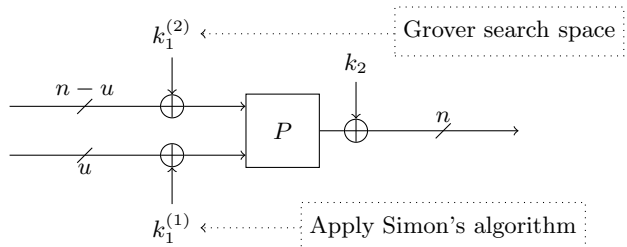
\begin{figure}
\centering
\begin{tikzpicture}[node distance=2em]
  \draw
  node at (0,0)[block,name=p,minimum size=3em] {$P$}
  node at (p.west) [name=pleft] {}
  node [name=pup, above of=pleft, node distance=1em] {}
  node [xor, name=xorup,circle,left of=pup,node distance=3em] {}
  node [l, above of=xorup,name=k12, above] {$k_1^{(2)}$}
  node [name=pdown, below of=pleft, node distance=1em] {}
  node [xor, name=xordown, circle, left of=pdown, node distance=3em] {}
  node [l, below of=xordown, name=k11,below] {$k_1^{(1)}$}
  node [left of=xorup, node distance = 7em,name=startup] {}
  node [left of=xordown, node distance = 7em,name=startdown] {};
  \draw[->] (xorup) -- (pup.center);
  \draw[->] (xordown) -- (pdown.center);
  \draw[->] (k12) -- (xorup);
  \draw[->] (k11) -- (xordown);
  \draw (startup) -- node {\midsection} node[above,node distance=1.5em] {$n-u$} (xorup);
  \draw (startdown) -- node {\midsection} node[below,node distance=2em] {$u$} (xordown);
  \draw
  node [xor, circle,name=xorend, right of=p, node distance=3em] {}
  node [l, above of=xorend,name=k2,above] {$k_2$}
  node [right of=xorend, node distance=7em, name=end] {};
  \draw (p) -- (xorend);
  \draw[->] (k2) -- (xorend);
  \draw[->] (xorend) -- node {\midsection} node[name=n,above, node distance=1.5em] {$n$} (end);
  
  \draw 
  node [block, dotted,name=grover, above of=n, node distance=3.5em,right=-2em] {Grover search space}
  node [block, dotted,name=simon, below of=n, node distance=4.5em] {Apply Simon's algorithm};
  \draw[dotted, ->] (grover.west|-k12) -- (k12);
  \draw[dotted,->] (simon.west|-k11) -- (k11);
  
 \end{tikzpicture}
\caption{Idea of our Q1 attack on the Even-Mansour construction.}
\label{fig:Q1EM}
\end{figure}
\paragraph{Attack description.}
Here we give the description of our Q1 attack.
Let $u$ be an integer such that $0 \leq u \leq n$.
Define $F : \{0,1\}^{n-u} \times \{0,1\}^u \rightarrow \{0,1\}^n$ by
\begin{equation}
F(i,x) = P(x\|i),
\end{equation}
and define $g : \{0,1\}^u \rightarrow \{0,1\}^n$ by $g(x) = E_{k_1,k_2}(x\|0^{n-u}).$

Note that $F(k^{(2)}_1,x) \oplus g(x)$ has the period $k^{(1)}_1$ since $F(k^{(2)}_1,x) \oplus g(x) = P(x\|k^{(2)}_1) \oplus P((x\oplus k^{(1)}_1) \| k^{(2)}_1) \oplus k_2$.
Our attack is described as the following procedure:
\begin{enumerate}
\item Run $\AlgQone$ for the above $F$ and $g$ to recover $k^{(2)}_1$.
\item Recover $k^{(1)}_1$ by applying $\SimQone$ to $f_{k^{(2)}_1}$ and $g$.
\item Compute $k_2 = E_{k_1,k_2}(0^n) \oplus P(k_1)$.
\end{enumerate}

\paragraph{Analysis.}
Below we assume that $u$ is not too small, \emph{e.g.}, $u \geq n / \log_2 n$.
This assumption is not an essential restriction since, if $u$ is too small, then the complexity of the first step becomes almost the same as the Grover search on $k_1$, which is not of interest.

If $P$ is a random permutation, we can assume that $f_i \oplus g = P( \cdot || i) \oplus E_{k_1,k_2}(\cdot || 0^{n-u})$ is far from periodic for all $i \neq k^{(2)}_1$, and that assumption~\eqref{eq:simoncondition} in Problem~\ref{problem:search-shift} holds.

Hence, by Proposition~\ref{prop:AlgQ1}, $\AlgQone$ in Step 1 recovers $k^{(2)}_1$ with a high probability by making $\bigO{2^u}$ classical queries to $g$ and the offline computation of $\AlgQone$ runs in time $\bigO{n^32^{(n-u)/2}}$.
Here, notice that one evaluation of $g$ (resp. $F$) can be done in $\bigO{1}$ evaluations of $E_{k_1,k_2}$ (resp. $P$).
In addition, from Proposition~\ref{prop:SimQ1}, $\SimQone$ in Step 2 recovers $k^{(1)}_1$ with a high probability by making $\bigO{2^u}$ classical queries to $g$ and the offline computation of $\AlgQone$ runs in time $\bigO{n^3}$.
Step 3 requires $\bigO{1}$ queries to $E_{k_1,k_2}$ and $\bigO{1}$ offline computations.

In summary, our attack recovers keys of the Even-Mansour construction with a high probability by making $D=\bigO{2^u}$ classical queries to $E_{k_1,k_2}$ and doing $T=\bigO{n^3 2^{(n-u)/2}}$ offline computations, which balances at $T = D = \tilde{\mathcal{O}}(2^{n/3})$.
By construction of $\AlgQone$ and $\SimQone$, our attack uses $\mathsf{poly}(n)$ qubits and $\mathsf{poly}(n)$ classical memory.

\subsubsection{Applications to concrete instances.}
The Even-Mansour construction is a commonly used cryptographic construction. The masks used in Even-Mansour are often derived from a smaller key, which can make a direct key-recovery using Grover's algorithm more efficient. This is for example the case in the CAESAR candidate Minalpher~\cite{Minalpher}. In general, we need to have a secret key of at least two thirds of the state size for our attack to beat the exhaustive search.

The Farfalle construction~\cite{DBLP:journals/tosc/BertoniDHPAK17} degenerates to an Even-Mansour construction if the input message is only 1 block long. Instances of this construction use variable states and key sizes. The Kravatte instance~\cite{DBLP:journals/tosc/BertoniDHPAK17} has a state size of 1600 bits, and a key size between 256 and 320 bits, which leads to an attack at a whopping cost of $2^{533}$ data and time, while the direct key exhaustive seach would cost at most $2^{160}$. Xoofff~\cite{DBLP:journals/tosc/DaemenHAK18} has a state size of 384 bits and a key size between 192 and 384 bits. Our attack needs $2^{128}$ data, which is exactly the data limit of Xoofff. Hence, it is relevant if the key size is greater than 256.

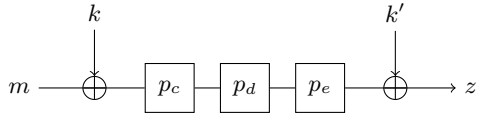
\begin{figure}[h]
  \centering
 \begin{tikzpicture}
  \draw
  node at (0,0)[l,name=m] {$m$}
  node [xor,circle, name=xor1, right of=m] {}
  node [block,name=f1, right of=xor1] {$p_c$}
  node [block,name=f2, right of=f1] {$p_d$}
  node [block,name=f3, right of=f2] {$p_e$}
  node [xor,circle, name=xor2, right of=f3] {}
  node [name=out, right of=xor2] {$z$}
  node [name=k1, above of=xor1] {$k$}
  node [name=k2, above of=xor2] {$k'$};
  \draw[->] (m) -- (f1) -- (f2) -- (f3) -- (out);
  \draw[->] (k1) -- (xor1);
  \draw[->] (k2) -- (xor2);
 \end{tikzpicture}
\caption{One-block Farfalle.}
\end{figure}

\subsection{Tradeoffs for the FX Construction}
The FX construction~\cite{DBLP:conf/crypto/KilianR96}  $\mathrm{FX}_{k,k_{in},k_{out}}$,  computes a ciphertext $c$ from a plaintext $p$ by $c \leftarrow E_{k}(p \oplus k_{in}) \oplus k_{out}$, where $E$ is a block cipher, $k$ is an $m$-bit key and $k_{in},k_{out}$ are two $n$-bit keys.
In the classical setting, there exists a classical attack with tradeoff $TD=2^{n+m}$, which balances at $T=D=2^{(n+m)/2}$ (see, for example, \cite{DBLP:conf/eurocrypt/Dinur15} for more details and memory trade-offs).

\subsubsection{Previous Q1 attacks on the FX construction.} 

Applying Grover as we did before on Even-Mansour on the keys $k_{in}$ and $k$, we can recover the keys with only two pairs of plaintext-ciphertext and a time complexity of $2^{(n+m)/2}$, while only needing a polynomial number of qubits. 

In~\cite{DBLP:conf/ctrsa/HosoyamadaS18}, Hosoyamada and Sasaki proposed a tradeoff $D \cdot T^6=2^{3(n+m)}$ for $D \leq \min\{2^n,2^{3(n+m)/7}\}$ with a polynomial amount of qubits, by using the multi-target preimage search by Chailloux et al~\cite{DBLP:conf/asiacrypt/ChaillouxNS17}.
The balance occurs at $D=T=2^{3(n+m)/7}$ (if $m \leq 4n/3$), which is smaller than the classical balanced point $2^{(n+m)/2}$.
The attack requires $M=D^{1/3}$ classical memory, thus the attack still requires exponentially large space at the balanced point.
This was the only Q1 attack with time $T \ll 2^{(n+m)/2}$ and a polynomial amount of qubits.

\subsubsection{Application of $\AlgQone$.}
We explain how to apply our algorithm $\AlgQone$ to the FX construction.
Our new tradoff is $T^2 \cdot D = 2^{n+m}$ for $D \leq 2^n$, which balances at $T=D=2^{(n+m)/3}$ (for $m \leq 2n$), using only $\mathsf{poly}(n)$ qubits and $\mathsf{poly}(n)$ classical memory.
See Table~\ref{table:tradeoffFX} for comparison of attack complexities under the condition that only $\mathsf{poly}(n)$ qubits are available.
\begin{table}[htb]
\centering
\begin{tabular}{|c||c|c|c|c|c}
\hline
Reference &
Classical attack & Grover & \cite{DBLP:conf/ctrsa/HosoyamadaS18} & \bf{[Ours]} \\
\toprule
\begin{tabular}{c}
Tradeoff of \\ $D$ and $T$
\end{tabular}
&
\begin{tabular}{c}
\\
$D \cdot T=2^{n+m}$ \\
$(D \leq 2^n)$
\end{tabular}
&
\begin{tabular}{c}
\\
$T=2^{(n+m)/2}$ \\
 $D =\mathrm{constant}$
\end{tabular}
&
\begin{tabular}{c}
\\
$D \cdot T^6=2^{3(n+m)}$ \\
$(D \leq \min\{2^n,2^{3n/7}\})$
\end{tabular}
&
\begin{tabular}{c}
\\
$D \cdot T^2=2^{n+m}$ \\
$(D \leq 2^n)$
\end{tabular} \\
\midrule
\begin{tabular}{c}
Num. of qubits
\end{tabular}
& - & $\mathsf{poly}(n)$ & $\mathsf{poly}(n)$ & $\mathsf{poly}(n)$ \\
\midrule
\begin{tabular}{c}
Class. memory
\end{tabular}
& $D$  & ${\mathsf{poly}(n)}$ & $D^{1/3}$ & ${\mathsf{poly}(n)}$ \\
\midrule
\begin{tabular}{c}
Balanced point \\ of $D$ and $T$
\end{tabular}
&
\begin{tabular}{c}
$2^{(n+m)/2}$ \\ ($m \leq n$)
\end{tabular}
& -&
\begin{tabular}{c}
$2^{3(n+m)/7}$ \\ ($m \leq 4n/3$)
\end{tabular}
& 
\begin{tabular}{c}
$2^{(n+m)/3}$ \\ ($m \leq 2n$)
\end{tabular}\\
\bottomrule
\end{tabular}
\medskip
\caption{Tradeoffs for Q1 attacks on the FX construction.
In this table we omit to write order notations, and ignore polynomial factors in the first and last rows.
}
\label{table:tradeoffFX}
\end{table}
\paragraph{Attack idea.}
Recall that, in the Q1 attack on the Even-Mansour construction in Section~\ref{sec:EMQ1}, we divided the first key $k_1$ to two parts $k^{(1)}_1$ and $k^{(2)}_1$ and applied Simon's algorithm to $k^{(1)}_1$ while we performed Grover search on $k^{(2)}_1$.

In a similar manner, for the FX construction $\FX_{k,k_{in},k_{out}}$ we divide the $n$-bit key $k_{in}$ in $k^{(1)}_{in}$ of $u$ bits and $k^{(2)}_{in}$ of $(n-u)$ bits. We apply Simon's algorithm to recover $k^{(1)}_{in}$ while we perform Grover search on $k$ in addition to $k^{(2)}_{in}$ (see Fig.~\ref{fig:Q1FX}).

\begin{figure}[htb]
\centering
\begin{tikzpicture}[node distance=2em]
  \draw
  node at (0,0)[block,name=p,minimum size=3em] {$E$}
  node at (p.west) [name=pleft] {}
  node [name=pup, above of=pleft, node distance=1em] {}
  node [xor, name=xorup,circle,left of=pup,node distance=3em] {}
  node [l, above of=xorup,name=k12,above] {$k_{in}^{(2)}$}
  node [name=pdown, below of=pleft, node distance=1em] {}
  node [xor, name=xordown, circle, left of=pdown, node distance=3em] {}
  node [l, below of=xordown, name=k11, below] {$k_{in}^{(1)}$}
  node [left of=xorup, node distance = 7em,name=startup] {}
  node [left of=xordown, node distance = 7em,name=startdown] {};
  \draw[->] (xorup) -- (pup.center);
  \draw[->] (k12) -- (xorup);
  \draw[->] (xordown) -- (pdown.center);
  \draw[->] (k11) -- (xordown);
  \draw (startup) -- node {\midsection} node[above,node distance=1.5em] {$n-u$} (xorup);
  \draw (startdown) -- node {\midsection} node[below,node distance=2em] {$u$} (xordown);
  \draw
  node [xor, circle,name=xorend, right of=p, node distance=3em] {}
  node [l, above of=xorend,name=k2,above] {$k_{out}$}
  node [right of=xorend, node distance=7em, name=end] {};
  \draw[->] (p) -- (xorend)-- node {\midsection} node[name=n,above, node distance=1.5em] {$n$} (end);
  \draw[->] (k2) -- (xorend);
  \draw 
  node [block, dotted,name=grover, above of=n, node distance=6em] {Grover search space}
  node [block, dotted,name=simon, below of=n, node distance=4.5em] {Apply Simon's algorithm};

  \draw[dotted,->] (simon.west|-k11) -- (k11);
  \draw[fill=black,draw opacity=0] (p.north west) rectangle ([yshift=-0.15cm]p.north east);
  \draw  
  node [l,name=k,above of=p, node distance=4.5em] {$k$};
  \draw[->] (k) -- node {\midsection} node [name=m,right,node distance=1.5em] {$m$} (p);
  \draw  node [name=tmp,above of=k] {};
  \draw[dotted, ->] (grover) -| (k12);
  \draw[dotted, ->] (grover) |- (k);
 \end{tikzpicture}
\caption{Idea of our Q1 attack on the FX construction.}
\label{fig:Q1FX}
\end{figure}
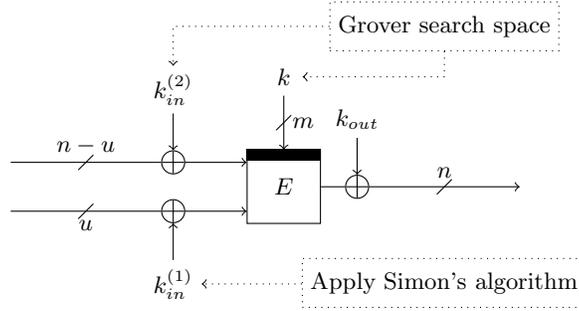

Then, roughly speaking, by applying $\AlgQone$ we can recover the key by making $D=2^u$ classical queries and $T=2^{(n-u)/2}$ offline Grover iterations (note that the offline computation cost for the Simon's algorithm is $\mathsf{poly}(n)$ and we ignore polynomial factors here for simplicity), which yields the tradeoff $D \cdot T^2 = 2^{(n+m)}$ for $D \leq 2^n$, with only $\mathsf{poly}(n)$ qubits and $\mathsf{poly}(n)$ classical memories.

\paragraph{Attack description.}
Here we give the description of our Q1 attack.
Let $u$ be an integer such that $0 \leq u \leq n$.
Define $F : \{0,1\}^{m+(n-u)} \times \{0,1\}^u \rightarrow \{0,1\}^n$ by
\begin{equation}
F(i\|j,x) = E_i(x\|j) (i \in \{0,1\}^{m}, j \in\{0,1\}^{n-u}),
\end{equation}
and define $g : \{0,1\}^u \rightarrow \{0,1\}^n$ by $g(x) = \FX_{k,k_{in},k_{out}}(x\|0^{n-u}).$

Note that $F(k \| k^{(2)}_{in},x) \oplus g(x)$ has the period $k^{(1)}_{in}$ since $F(k \| k^{(2)}_{in},x) \oplus g(x) = E_k(x\|k^{(2)}_{in}) \oplus E_k((x\oplus k^{(1)}_{in}) \| k^{(2)}_{in}) \oplus k_{out}$.
Our attack procedure runs as follows:
\begin{enumerate}
\item Run $\AlgQone$ for the above $F$ and $g$ to recover $k$ and $k^{(2)}_{in}$.
\item Recover $k^{(1)}_{in}$ by applying $\SimQone$ to $f_{k \| k^{(2)}_{in}}$ and $g$.
\item Compute $k_{out} = \FX_{k,k_{in},k_{out}}(0^n) \oplus E_k(k_{in})$.
\end{enumerate}

\paragraph{Analysis.}
We assume that $m = \bigO{n}$, which is the case for usual block ciphers.
In the same way as in the analysis for the attack on the Even-Mansour construction in Section~\ref{sec:EMQ1}, if $E$ is a (pseudo) random permutation family and $u$ is not too small (\emph{e.g.} $u \geq n / \log_2 n$), we observe that the assumption~\eqref{eq:simoncondition}, rephrased as:
\begin{align}\label{eq:AssumptionQ1FXequiv}
&\max_{t \in \{0,1\}^n \setminus \{0^n\}} \Pr_{x \in \{0,1\}^n} \left[ E_i\left(x||j\right) \oplus E_i\left( x \oplus t || j \right) \oplus E_k\left( x \oplus k^{(1)}_{in} ||  k^{(2)}_{in} \right) \right. \nonumber \\
& \hspace{100pt} \left. \oplus E_k\left( x \oplus t \oplus k^{(1)}_{in} || k^{(2)}_{in} \right) = 0 \right] \leq 1/2
\end{align}
holds for all $(i,j) \neq (k,k^{(2)}_1)$ with overwhelming probability.
This again implies that the claims of Propositions~\ref{prop:AlgQ1} and~\ref{prop:SimQ1} hold for $\AlgQone$ in Step 1 and $\SimQone$ in Step 2, respectively.

Thus our attack recovers keys of the FX construction with a high probability by making $D=\bigO{2^u}$ classical queries to $\FX_{k,k_{in},k_{out}}$ and doing $T=\bigO{n^32^{(m+n-u)/2}}$ offline computations for $D \leq 2^n$, which balances at $T = D = \tilde{\mathcal{O}}(2^{(n+m)/3})$ if $m \leq 2n$.
Our attack uses only $\mathsf{poly}(n)$ qubits and $\mathsf{poly}(n)$ classical memory by construction of $\AlgQone$ and $\SimQone$.

\subsubsection{Application to concrete instances.} 
DESX~\cite{DBLP:conf/crypto/KilianR96} has a 64-bit state, two 64-bit whitening key and one
56-bit inner key. From Propositions~\ref{prop:AlgQ2} and~\ref{prop:cval}, we can estimate that our attack needs roughly $2^{42}$ classical queries and $2^{40}$ quantum computations of the cipher circuit. 

PRINCE~\cite{DBLP:conf/asiacrypt/BorghoffCGKKKLNPRRTY12}, and PRIDE~\cite{DBLP:conf/crypto/AlbrechtDKLPY14} are two ciphers using the FX construction with a 64-bit state, a 64-bit inner key and two 64-bit whitening keys. Hence, from Propositions~\ref{prop:AlgQ2} and~\ref{prop:cval}, we can estimate that our attack needs roughly $2^{45}$ quantum queries and $2^{43}$ quantum computations of the cipher circuit.

We can also see some encryption modes as an instance of the FX construction. This is for example the case of the XTS mode~\cite{DBLP:journals/ieeesp/Martin10}, popular for disk encryption. It is generally used with AES-256 and two whitening keys that
depend on the block number and another 256-bit key. Hence, with the full codebook of one block, we can obtain the first key and the value of the whitening keys of the corresponding block. Once the first key is known, the second can easily be brute-forced from a few known plaintext-ciphertext couples in other blocks.

Adiantum~\cite{DBLP:journals/tosc/CrowleyB18} is another mode for disk encryption that uses a variant of the FX construction with AES-256 and Chacha. There is however one slight difference: the FX masking keys are added with a modular addition instead of a xor. The FX construction is still vulnerable~\cite{DBLP:conf/asiacrypt/BonnetainN18}, but we will need to use Kuperberg's algorithm~\cite{DBLP:conf/tqc/Kuperberg13} instead of Simon's algorthm.  As before, with the full codebook on one block, we can recover the AES and Chacha keys in a time slightly larger than $2^{256}$.

\subsection{Other Applications}

\subsubsection{Chaskey.} The lightweight MAC Chaskey~\cite{DBLP:conf/sacrypt/MouhaMHWPV14} is very close to an Even-Mansour construction (see Figure~\ref{fig:chaskey}). Since the last message block ($m_2$ in Figure~\ref{fig:chaskey}) is XORed to the key $K_1$, we can immediately apply our Q1 attack and recover $K_1$ and the value of the state before the xoring of the last message block. As $\pi$ is a permutation easy to invert, this allows to retrieve $K$. The Chaskey round function applies on 128 bits. It contains 16 rounds with 4 modular additions on 32 bits, 4 XORs on 32 bits and some rotations. With a data limit of $2^{48}$, as advocated in the specification, our attack would require roughly $2^{(128 - 48)/2} \times 2^{19} = 2^{59}$ quantum gates, where the dominant cost is solving the 80-dimensional linear system inside each iteration of Grover's algorithm.

\begin{figure}
  \centering
 \begin{tikzpicture}
  \draw
  node at (0,0)[l,name=m] {$K$}
  node [xor,circle, name=xor1, right of=m] {}
  node [sponge,name=f1, right of=xor1] {$\pi$}
  node [xor,circle, name=xor2, right of=f1] {}
  node [xor,circle, name=xor3, right of=xor2] {}
  node [sponge,name=f2, right of=xor3] {$\pi$}
  node [xor,circle, name=xor4, right of=f2] {}
  node [name=k1, above of=xor1] {$m_1$}
  node [name=k2, above of=xor2] {$m_2$}
  node [name=k3, above of=xor3] {$K_1$}
  node [name=k4, above of=xor4] {$K_1$}
  node [block, name=trunk, right of=xor4] {$\text{Trunk}_t$}
  node [name=out, right of=trunk, right] {Tag};
  \draw[->] (m) -- node {\midsection} node [above=0.2em] {$128$} (xor1) -- (f1) -- (f2) -- (trunk) -- (out);
  \draw[->] (k1) -- (xor1);
  \draw[->] (k2) -- (xor2);
  \draw[->] (k3) -- (xor3);
  \draw[->] (k4) -- (xor4);
 \end{tikzpicture}
\caption{Two-block Chaskey example.}
\label{fig:chaskey}
\end{figure}
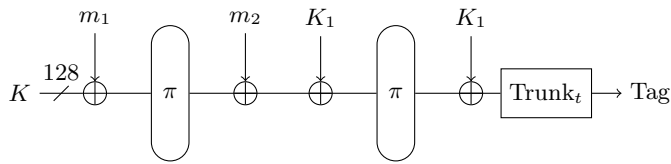

\subsubsection{Sponges}

Our attack can be used on sponges if there is an input injected on a fixed state. In general, it has two drawbacks: the nonce has to be fixed, and the cost of the attack is at least $2^{c/2}$ with $c$ the capacity of the sponge, which is often
the classical security parameter. However, there are some cases where our attack is of interest.

In particular, our attack needs a set of values that contains an affine space. If a nonce was injected the same way the messages are, then we only need to know the encryptions of identical messages, with a set of nonces that fills an affine space. Nonce-respecting
adversaries are generally allowed to choose the nonce, but here, the mere assumption that the nonce is incremented for each message (which is the standard way nonces are processed in practice) is sufficient: A set of $2^k$ consecutive values contains an affine space
of $(\mathbb{Z}/(2))^{k-1}$.

This is the case in the Beetle mode of lightweight authenticated encryption~\cite{DBLP:journals/tches/ChakrabortiDNY18}, whose initialization phase is described as $(K_1 \oplus N) \| K_2 \mapsto f((K_1 \oplus N) \| K_2)$, where $K_1, N \in \{0,1\}^r$, $K_2 \in \{0,1\}^c$, and $f$ is a $(r+c)$-bit permutation.

Here, the nonce is directly added to the key $K_1$, but as the key has the same length as the state, the attack would still work if the nonce was added after the first permutation.
In Beetle[Light+], the rate is $r = 64$ bits and the capacity $c = 80$ bits. The rate is sufficiently large to embed 48 varying bits for the nonce; in that case, by making $2^{48}$ classical queries and $2^{48}$ Grover iterations, we can recover the secret $K_1 || K_2$. In Beetle[Secure+], $r = c = 128$ bits. We can recover $K_1 || K_2$ with $2^{85}$ messages and Grover iterations.

\begin{figure}
 \centering
\begin{tikzpicture}
 \draw
 node at (0,0)[l,name=kmid] {}
 node [sponge,name=f1, right of=kmid,node distance=6em] {$f$}
 node [name=f1in,left of=f1, node distance=1em] {}
 node [l,name=k1, above of=kmid, node distance=1em] {$K_1 \oplus N$}
 node [l,name=k2, below of=kmid, node distance=1em] {$K_2$}
 node [name=f1out, right of=f1,node distance=1em] {}
 node [name=outup, above of=f1out, node distance=1em] {}
 node [name=outdown, below of=f1out, node distance=1em] {}
 node [name=rate, right of=outup,node distance=5em] {}
 node [name=capacity, right of=outdown, node distance=5em] {};
 \draw[->] (k1) -- (k1-|f1.west);
 \draw[->] (k2) -- (k2-|f1.west);
 \draw[->] (f1.east|-rate) -- node {\midsection} node[above, node distance=1.5em] {$r$} (rate);
 \draw[->] (f1.east|-capacity) -- node {\midsection} node[below, node distance=1.5em] {$c$} (capacity);
\end{tikzpicture}
\caption{Beetle state initialization.}
\label{fig:beetle}
\end{figure}
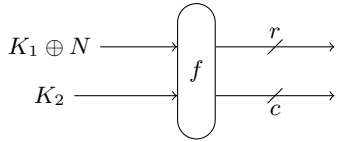

%% file: discussion.tex

In this section, we discuss on the application of our attack idea to related-key attacks, to some slide attacks, and to an extension of Problem~\ref{problem:search-shift}.
See also Section~B in the \acversion{full version of the paper~\cite{DBLP:journals/iacr/BonnetainHNSS19}}\fullversion{Appendix} for discussions on adaptive attacks and non-adaptive attacks.

\subsection{Related Keys}

Consider a block cipher $E_k$ with a key and block size of $n$ bits. In the related-key setting, as introduced in~\cite{DBLP:journals/cryptologia/WinternitzH87}, we are not only allowed to make chosen plaintext or ciphertext queries to a secret-key oracle hiding $k$, but also to query $E_{k \oplus \ell} (m)$ for any $n$-bit difference $\ell$ and message $m$. Classically, this is a very powerful model, but it becomes especially meaningful when the block cipher is used inside a mode of operation (\emph{e.g.} a hash function) in which key differences can be controlled by the attacker. It is shown in~\cite{DBLP:journals/cryptologia/WinternitzH87} that a secret key recovery in this model can be performed in $2^{n/2}$ operations, as it amounts to find a collision between some query $E_{k\oplus \ell} (m)$ and some offline computation $E_{\ell'}(m)$ (we can use more than a single plaintext $m$ to ensure an overwhelming success probability).

R{\"{o}}tteler and Steinwandt~\cite{DBLP:journals/ipl/RottelerS15} noticed that, if a quantum adversary has \emph{superposition} access to the oracle that maps $\ell$ to $E_{k \oplus \ell} (m)$, it can mount a key-recovery in polynomial time using Simon's algorithm. Indeed, one can define a function:
\[
f(x) = E_{k \oplus x}(m) \oplus E_x(m)
\]
which has $k$ as hidden period, apply Simon's algorithm and recover $k$. This attack works for any block cipher, even ideal. In contrast, in the Q2 quantum attacker model, we know that some constructions are broken, but it does not seem to be the case for all of them.

With our algorithm $\AlgQone$, we are able to translate this related-key superposition attack into an attack where the related-key oracle is queried only classically, but the attacker has quantum computing power. We write $k = k_1 || k_2$ where $k_1$ has $n/3$ bits and $k_2$ has $2n/3$ bits. We query $E_{(k_1 || k_2) \oplus (\ell_1 || 0)} (m)$ for a fixed $m$ and all $n/3$-bit differences $\ell_1$. Then we perform a Grover search on $k_2$. The classical security level in presence of a related-key oracle of this form, which is $2^{n/2}$, is reduced quantumly to $2^{n/3}$. This shows that the transition to a quantum setting has an impact on the related-key security even if the oracle remains classical.

As a consequence, we could complete the security claims of the $16$-round version of the block cipher \textsc{Saturnin}~\cite{saturnin}, a submission to the ongoing NIST lightweight cryptography competition\footnote{\url{https://csrc.nist.gov/Projects/Lightweight-Cryptography}}. The authors of \textsc{Saturnin} gave security claims against quantum attackers meeting the best generic attacks. No claims were given regarding the $Q_1$ model for related-key attacks.  Our result gives the best generic quantum related-key attack on ideal block ciphers without superposition queries, and sets the level of security that should be expected from a block cipher in this setting: the key can be recovered in quantum time $\bigOt{2^{n/3}}$ for a block cipher of $n$ bits (and using $2^{n/3}$ classical related-key queries). The corresponding security level for \textsc{Saturnin}$_{16}$, which has blocks of $256$ bits, lies at $2^{256/3} = 2^{85}$: we can say that in the $Q_1$ related-key setting, \textsc{Saturnin}$_{16}$ should have no attack with time complexity lower than $2^{85}$.


\subsection{Slide Attacks}

Quantum slide attacks are a very efficient quantum counterpart of the classical slide attacks~\cite{DBLP:conf/fse/BiryukovW99}. They have been introduced in~\cite{DBLP:conf/crypto/KaplanLLN16}, with a polynomial-time attack on 1-round self-similar ciphers. In many cases, our algorithm does not improve these attacks, because they are already too efficient and do not rely on a partial exhaustive search. Still, some of them use a partial exhaustive search. This is the case of the slide attack against 2 round self-similar  ciphers of~\cite{DBLP:conf/asiacrypt/Leander017} and the slide attacks against whitened Feistels of~\cite{DBLP:journals/iacr/BonnetainNS18}. 

For example, we can see a 2 round self-similar cipher as an example of iterated FX cipher, as in Figure~\ref{fig:iFX}.
Define functions $p_i$, $F_i$, and $g$ as
\[
  p_i((b,x))=\left\{\begin{matrix} (0,E_i(x)) &\text{if $b = 0$} \\ (1,x) & \text{if $b = 1$}\end{matrix}\right.,\quad F_i((b,x), y) = \left\{\begin{matrix} y\oplus x &\text{if $b = 0$} \\ E_i(y) \oplus x & \text{if $b = 1$}\end{matrix}\right.,
\]
and $g((b,x)) = \text{iFX}(x)$. We have the property that $\text{iFX}(E_{k_2}(x\oplus k_1))\oplus (x\oplus k_1) = E_{k_2}(\text{iFX}(x))\oplus x$. Hence, we have the hidden period $(1,k_1)$ in the function $f_{k_2}((b,x)) = F_{k_2}\left((b,x),g(p_{k_2}(b,x))\right)$. 
To apply our attack, we need to compute
$\sum_{x,b}\ket{x}\ket{b}\ket{f_i((b,x))}$ from the state $\sum_x\ket{x}\ket{\text{iFX}(x)}$. We first need to add one qubit to obtain $\sum_x\ket{x}(\ket0+\ket1)\ket{\text{iFX}(x)}$. Then, conditioned on the second register to be 0, we transform $x$ into $E_i^{-1}(x)$. Next, conditioned on the second register to be 1, we transform $\text{iFX}(x)$ into $E_i(\text{iFX}(x))$. Finally, we add the first register to the third. Hence, we can apply our attack, and retrieve $k_1$ and $k_2$ using
$\bigO{|k_1|}$ queries and $\bigO{|k_1|^32^{|k_2|/2}}$ time, assuming $|k_1| = \Omega(|k_2|)$.
\begin{figure}[h]
  \centering
 \begin{tikzpicture}
  \draw
  node at (0,0)[l,name=m] {$m$}
  node [xor,circle, name=xor1, right of=m] {}
  node [block,name=f1, right of=xor1] {$E_{k_2}$}
  node [xor,circle, name=xor2, right of=f1] {}
  node [name=d, right of=xor2] {...}
  node [xor,circle, name=xor3, right of=d] {}
  node [block,name=f2, right of=xor3] {$E_{k_2}$}
  node [xor,circle, name=xor4, right of=f2] {}
  node [name=k1, above of=xor1] {$k_1$}
  node [name=k2, above of=xor2] {$k_1$}
  node [name=k3, above of=xor3] {$k_1$}
  node [name=k4, above of=xor4] {$k_1$}
  node [name=out, right of=xor4, right] {$\text{iFX}(m)$};
  \draw[->] (m) --  (xor1) -- (f1) -- (d) -- (f2) -- (out);
  \draw[->] (k1) -- (xor1);
  \draw[->] (k2) -- (xor2);
  \draw[->] (k3) -- (xor3);
  \draw[->] (k4) -- (xor4);
 \end{tikzpicture}
\caption{Iterated-FX cipher.}
\label{fig:iFX}
\end{figure}
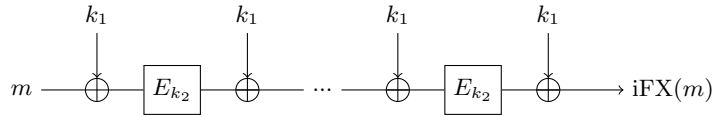

The above problem of recovering keys can be generalized as the following problem, which can be solved by the same strategy as above.
\begin{problem}[Constructing and Finding a Hidden Period]\label{problem:asymmetric-shift}
Let $g~:\zo^n \rightarrow \zo^\ell$ be a function, $i\in I$,  $p_i : \zo^n \to \zo^n$ be a permutation and $F_i : \zo^n \times \zo^\ell \to \zo^\ell$  be a function such that $F_i(x,\cdot)$ is a permutation.
Assume that there exists $i_0 \in I$ such that $f_{i_0}(x) = F_{i_0}\left(x,g(p_{i_0}(x))\right)$ has a period, \emph{i.e.}: $\forall x \in \zo^n, f_{i_0}(x) = f_{i_0}(x\oplus s)$ for some $s$. Assume that we are given quantum oracle access to $F_i$ and $p_i$ and
classical or quantum oracle access to $g$.
(In the Q1 setting, $g$ will be a classical oracle. In the Q2 setting, $g$ will be a quantum oracle.)
Then find $i_0$ and $s$.
\end{problem} 

This problem assumes that $g$ is a keyed function, and that we can reversibly transform $(x,g(x))$ into a couple $(y,f_i(y))$, with $f_i$ a periodic function if $i = i_0$. We can see this transformation as a generalization of the CCZ equivalence~\cite{add:CarChaZin98}, where the function mapping the graph
of $g$ and the graph of $f_i$ do not need to be an affine function.
There may also be more than one solution (in which case we just want to find one), or there may be none, just as Grover's algorithm can handle cases with many expected solutions, or distinguish whether there is a solution or not.
Note that Problem~\ref{problem:search-shift} is a special case of the above problem, in the case where $p_i$ is the identity, and $F_i$ is only the xoring of $g$ and another function.

%% file: erroranalyses.tex

\newpage

\section{Proofs for Propositions \ref{prop:AlgQ2} and \ref{prop:AlgQ1}}\label{sec:erroranalyses}
Here we give error and complexity analyses for $\AlgQone$ and $\AlgQtwo$, and give proofs for Propositions \ref{prop:AlgQ2} and \ref{prop:AlgQ1}.

Recall that, in both of $\AlgQone$ and $\AlgQtwo$, we run a testing procedure $\sf test$ in each iteration of the Grover search.
For each $i \in \{0,1\}^m$ ($\{0,1\}^m$ is the search space), we expect that $\sf test$ checks whether $i$ is a good element by using auxiliary quantum data $\ket{\psi_g}$, without changing $\ket{\psi_g}$.
However, in fact $\sf test$ outputs results only with some errors, and $\ket{\psi_g}$ is slightly modified at each test.
Thus we have to do two kinds of analyses:
\begin{enumerate}
\item Analyses on the error of the procedure $\sf test$.
\item Analyses on how much the error of $\sf test$ affects the success probability of the entire Grover search.
\end{enumerate}

In Section~\ref{sec:ErrorSimon}, we give error analyses of the testing procedure $\sf test$.
In Section~\ref{sec:ErrorQAA}, we analyze how the error of $\sf test$ affects the success probability of the Grover search, in a general setting.
Actually we discuss about the quantum amplitude amplification (QAA) technique, which is a generalization of the entire Grover search.
Finally, in Section~\ref{sec:ErrorFin}, we prove Propositions~\ref{prop:AlgQ2} and \ref{prop:AlgQ1} by using the results in Sections~\ref{sec:ErrorSimon} and \ref{sec:ErrorQAA}.

\subsection{Error Analyses for the testing procedure $\sf test$}\label{sec:ErrorSimon}

Recall that, for each function $g : \{0,1\}^n \rightarrow \{0,1\}^\ell$ and a fixed positive integer $c$, the quantum state $\ket{\psi_g}$ is defined as 
\begin{align}
\ket{\psi_{g}}
&:= \bigotimes^{cn} \sum_{x \in \{0,1\}^{n}} \ket{x} \ket{g(x)}
\end{align}
Roughly speaking, we expect that the unitary operator $\sf test$ satisfies
\begin{equation}
{\sf test} \ket{i}\ket{\psi_{g}} \ket{b} =
\begin{cases}
\ket{i}\ket{\psi_{g}} \ket{b} & \text{ if $(f_i \oplus g)$ does not have a period,}\\
\ket{i}\ket{\psi_{g}} \ket{b \oplus 1} & \text{ if $(f_i \oplus g)$ has a period.}
\end{cases}
\end{equation}
holds with a small error.
Recall that $\sf test$ is defined as a unitary operator that realizes the following procedures:
\begin{enumerate}
\item Query $cn$ times to $F$ to obtain
\begin{align}
&\ket{i} \otimes \left( \sum_{x_1 \in \{0,1\}^{n}}  \ket{x_1} \ket{(f_i \oplus g)(x_1)} \right) \otimes \cdots \nonumber \\
&\hspace{50pt} \cdots \otimes
\left( \sum_{x_{cn} \in \{0,1\}^{n}}  \ket{x_{cn}} \ket{(f_i \oplus g)(x_{cn})} \right) \otimes \ket{b}.
\end{align}
\item Apply $I_m \otimes \left(H^{\otimes n} \otimes I_\ell \right)^{cn} \otimes I_1$ to obtain
\begin{align}
&\ket{i} \otimes \left( \sum_{u_1,x_1 \in \{0,1\}^{n}} (-1)^{u_1 \cdot x_1} \ket{u_1} \ket{(f_i \oplus g)(x_1)} \right) \otimes \cdots \nonumber \\
&\hspace{50pt} \cdots \otimes
\left( \sum_{u_{cn},x_{cn} \in \{0,1\}^{n}} (-1)^{u_{cn} \cdot x_{cn}} \ket{u_{cn}} \ket{(f_i \oplus g)(x_{cn})} \right) \otimes \ket{b}.
\end{align}
\item Compute $d := \dim({\rm Span}(u_1,\dots,u_{cn}))$, set $r:=0$ if $d = n$ and $r:=1$ if $d < n$, and add $r$ to $b$.
Then uncompute $d$ and $r$, and obtain
 \begin{align}
&\ket{i} \otimes \sum_{\substack{u_1, \dots, u_{cn} \\ x_1,\dots,x_{cn}}} (-1)^{u_1 \cdot x_1} \ket{u_1} \ket{(f_i \oplus g)(x_1)} \otimes \cdots \nonumber \\
&\hspace{50pt} \cdots \otimes (-1)^{u_{cn} \cdot x_{cn}} \ket{u_{cn}} \ket{(f_i \oplus g)(x_{cn})}  \otimes \ket{b \oplus r}.
\end{align}
\item Uncompute Steps 1 and 2.
\end{enumerate}
Note that if $(f_i \oplus g)$ has a period, then $r=1$ always holds and ${\sf test} \ket{i}\ket{\psi_{g}} \ket{b} = \ket{i}\ket{\psi_{g}} \ket{b \oplus 1}$ follows.
On the other hand, even if $(f_i \oplus g)$ does not have any period, ${\sf test}\ket{i} \ket{\psi_{g}} \ket{b}$ is not necessarily equal to $\ket{i}\ket{\psi_{g}} \ket{b}$.
However, we expect that ${\sf test}\ket{i} \ket{\psi_{g}} \ket{b} = \ket{i}\ket{\psi_{g}} \ket{b} + \ket{\delta}$ holds for some small error term $\ket{\delta}$.

Moreover, in our attack setting we apply the operator ${\sf test}$ to a superposition:
\begin{equation}\label{eq:UsimonGeneralState}
\sum_{\substack{f_i \in F_0}} \alpha_i \ket{i} \ket{\psi_{g}} \ket{b}
+
\sum_{f_i \in F_1} \beta_i \ket{i} \ket{\psi_{g}} \ket{b},
\end{equation}
where $\sum |\alpha_i|^2 + \sum |\beta_i|^2 = 1$, and $F_0,F_1$ are the subsets of $F$ defined as
\begin{align}
F_0 &:= \left\{ f_i \in F \middle| f_i \oplus g  \text{ does not have any period}\right\}, \nonumber \\
 F_1 &:= \left\{ f_i \in F \middle| f_i \oplus g \text{ has a period}\right\}. \nonumber
\end{align}
Roughly speaking, we expect that, when we apply ${\sf test}$ to the state~\eqref{eq:UsimonGeneralState}, it changes to
\begin{equation}
\sum_{f_i \in F_0} \alpha_i \ket{i} \ket{\psi_{g}} \ket{b}
+
\sum_{f_i \in F_1} \beta_i \ket{i} \ket{\psi_{g}} \ket{b \oplus 1}
\end{equation}
up to a small error term.

We can show the following lemma, which guarantees that the operator ${\sf test}$ works as we expect if $c$ is sufficiently large.

\begin{lemma}\label{lem:SimonErrorBound}
It holds that
\begin{align}
&{\sf test} \left( \sum_{f_i \in F_1} \beta_i \ket{i} \ket{\psi_{g}} \ket{b} \right) =
\sum_{f_i \in F_1} \beta_i \ket{i} \ket{\psi_{g}} \ket{b \oplus 1}
\label{eq:SimonErrorBoundG1}
\end{align}
for $b \in \{0,1\}$.
In addition, suppose that there exists a constant $0 \leq \epsilon < 1$ such that the following condition holds:
\begin{equation}\label{eq:SimonErrorHScondition}
\max_{ \substack{t \in \{0,1\}^{n} \setminus \{0^{n}\} \\ f_i \in F_0}} \Pr_{x \gets \{0,1\}^{n}}\left[ (f_i \oplus g)(x \oplus t)= (f_i \oplus g)(x) \right] \leq \epsilon
\end{equation}

Then, for arbitrary $b \in \{0,1\}$, there exists a vector $\ket{\delta_b}$ such that $\| \ket{\delta_b} \| < 2^{(n+1)/2} ((1 + \epsilon) / 2)^{cn/2}$ and
\begin{align}
&{\sf test} \left( \sum_{f_i \in F_0} \alpha_i \ket{i} \ket{\psi_{g}} \ket{b} \right)  =
\sum_{f_i \in F_0} \alpha_i \ket{i} \ket{\psi_{g}} \ket{b}
+
\ket{\delta_b} \label{eq:SimonErrorBoundG0}
\end{align}
holds.
\end{lemma}
\begin{proof}
\eqref{eq:SimonErrorBoundG1} obviously follows from the definition of ${\sf test}$ and $F_1$.
Below we show that \eqref{eq:SimonErrorBoundG0} holds. Notice that~\eqref{eq:SimonErrorHScondition} holds when the $f_i \in F_0$ are almost random and summarizes that each of them is far from periodic.

Let $U_1$, $U_2$, and $U_3$ be the unitary operators that realize Steps 1, 2 and 3 of ${\sf test}$, respectively.
($U^*_1$ and $U^*_2$ correspond to uncomputing $U_1$ and $U_2$, respectively, and ${\sf test} = U^*_1 U^*_2 U_3 U_2 U_1$ holds.)
Note that $U_1 \ket{i}\ket{\psi_g}\ket{b} = \ket{i}\ket{\psi_{f_i \oplus g}}\ket{b}$ holds for each $i$.

Let $\ket{\psi^{(i)}_2} := U_2U_1 \ket{i}\ket{\psi_{g}} \ket{b} = U_2 \ket{i}\ket{\psi_{f_i \oplus g}} \ket{b}$.
In addition, let us decompose $\ket{\psi^{(i)}_2}$ as
\begin{equation}
\ket{\psi^{(i)}_2} = \ket{i} \otimes \ket{\eta^{(i)}_{\good}}\otimes\ket{b} + \ket{i} \otimes \ket{\eta^{(i)}_{\bad}} \otimes \ket{b},
\end{equation}
where $\ket{\eta^{(i)}_{\good}}$ and $\ket{\eta^{(i)}_{\bad}}$ are the vectors projected to the subspaces that correspond to the tuple $(u_1,\dots,u_{cn})$ such that $\dim({\rm Span}(u_1,\dots,u_{cn})) = n$ and $\dim({\rm Span}(u_1,\dots,u_{cn})) < n$, respectively, in Step 2 of $\sf test$.
In particular,
 \begin{align}
\ket{\eta^{(i)}_{\bad}} &= \sum_{\substack{ x_1,\dots,x_{cn} \\ u_1, \dots, u_{cn} \\ \dim({\rm Span}(u_1,\dots,u_{cn})) < n}} (-1)^{u_1 \cdot x_1} \ket{u_1} \ket{(f_i \oplus g)(x_1)} \otimes \cdots \nonumber \\
& \hspace{75pt} \cdots \otimes (-1)^{u_{cn} \cdot x_{cn}} \ket{u_{cn}} \ket{(f_i \oplus g)(x_{cn})}
\end{align}
holds.
Then
\begin{equation}
U_3U_2U_1\ket{i}\ket{\psi_{g}}\ket{b} = U_3\ket{\psi^{(i)}_2} = \ket{i}\ket{\eta^{(i)}_{\good}} \ket{b} + \ket{i}\ket{\eta^{(i)}_{\bad}}\ket{b \oplus 1}
\end{equation}
and
\begin{align}
{\sf test}\ket{i}\ket{\psi_{g}}\ket{b} &= U^*_1U^*_2 U_3 U_2U_1\ket{i}\ket{\psi_{g}}\ket{b} \nonumber \\
&= U^*_1U^*_2\left(\ket{i} \otimes \ket{\eta^{(i)}_{\good}} \otimes \ket{b} + \ket{i}\otimes \ket{\eta^{(i)}_{\bad}}\ket{b}\right) \nonumber \\
&\quad + U^*_1U^*_2\left(\ket{i}\ket{\eta^{(i)}_{\bad}}\ket{b\oplus 1} - \ket{i}\ket{\eta^{(i)}_{\bad}}\ket{b} \right) \nonumber \\
&= \ket{i}\ket{\psi_{g}}\ket{b} + \ket{\delta_{b,i}}
\end{align}
follows, where $\ket{\delta_{b,i}}= U^*_1U^*_2\left(\ket{i}\ket{\eta^{(i)}_{\bad}}\ket{b\oplus 1} - \ket{i}\ket{\eta^{(i)}_{\bad}}\ket{b} \right)$.
Therefore we have
\begin{align}
&{\sf test} \left( \sum_{f_i \in F_0} \alpha_i \ket{i} \ket{\psi_{g}} \ket{b} \right)  =
\sum_{f_i \in F_0} \alpha_i \ket{i} \ket{\psi_{g}} \ket{b}
+
\ket{\delta_b},
\end{align}
where $\ket{\delta_b} = \sum_{f_i \in F_0} \alpha_i \ket{\delta_{b,i}}$. 
Note that
\begin{equation}
\| \ket{\delta_b} \| = \sqrt{2} \left\| \sum_{f_i \in F_0} \alpha_i \ket{i} \ket{\eta^{(i)}_\bad} \ket{b} \right\| 
\end{equation}
holds.

By orthogonality of the $\ket{i}$, we have:
\begin{equation}
\| \ket{\delta_b} \|^2 = 2 \sum_{f_i \in F_0} |\alpha_i|^2 \left\| \ket{\eta^{(i)}_\bad} \right\|^2 \leq 2 \sum_{f_i \in F_0} |\alpha_i|^2  \max_{f_i \in F_0}  \left\| \ket{\eta^{(i)}_\bad} \right\|^2 \leq 2 \max_{f_i \in F_0}  \left\| \ket{\eta^{(i)}_\bad} \right\|^2 \label{eq:SimonDeltaBound}
\end{equation}
which means that it suffices to focus on a single $f_i$ and bound $p_{\bad}^{(i)} := \left\| \ket{\eta^{(i)}_\bad} \right\|^2$. This is the probability of obtaining, for this function $f_i$, a tuple $(u_1,\dots,u_{cn})$ such that $\dim(\linebreak[3]{\rm Span}(u_1,\dots,u_{cn})) < n$ upon measurement of the state $\ket{\psi^{(i)}_2} = \ket{i} \otimes \left( \ket{\eta^{(i)}_{\good}} + \ket{\eta^{(i)}_{\bad}} \right) \otimes \ket{b}$. We just have to make sure that it is low for all $f_i$.

$\dim({\rm Span}(u_1,\dots,u_{cn})) < n$ holds if and only if there exists $t \in \{0,1\}^{n} \setminus \{0^{n}\}$ such that $t \perp u_j$ for all $1 \leq j \leq cn$.
In addition, the registers $\ket{u_1},\dots,\ket{u_{cn}}$ of $\ket{\psi^{(i)}_2}$ are unentangled from each other.
Thus we have that
\begin{align}
p_{\bad}^{(i)} = &\Pr\left[  \left(u_1,\dots,u_{cn}\right) \gets \left(\text{measure} \ket{\psi^{(i)}_2} \right) : \right. \nonumber \\ & \hspace{100pt} \left. \phantom{\ket{\psi^{(i)}_2}} \exists t \neq 0^{n} \text{ s.t. } t \perp u_j \text{ for } 1 \leq j \leq cn \right] \nonumber \\
&\leq \sum_{t \in \{0,1\}^{n} \setminus \{0^{n}\} }  \Pr\left[  \left(u_1,\dots,u_{cn}\right) \gets \left(\text{measure} \ket{\psi^{(i)}_2} \right) :  \right. \nonumber \\ & \hspace{100pt} \left.  \phantom{\ket{\psi^{(i)}_2}} t \perp u_j \text{ for } 1 \leq j \leq cn \right] \nonumber \\
&= \sum_{t \in \{0,1\}^{n} \setminus \{0^{n}\} }  \left( \Pr\left[ u \gets \left(\text{measure} \ket{\Psi^{(i)}}\right) :  t \perp u \right] \right)^{cn} \nonumber \\
&\leq 2^n   \left( \max_{t \in \{0,1\}^{n} \setminus \{0^{n}\} } \Pr\left[ u \gets \left(\text{measure} \ket{\Psi^{(i)}}\right) :  t \perp u \right] \right)^{cn} \label{eq:Psi1ProbBad}
\end{align}
holds, where $\ket{\Psi^{(i)}} = \sum_{u,x} (-1)^{u \cdot x} \ket{u} \ket{(f_i \oplus g)(x)}$.

Now we use the following claim as a fact, which is shown as a subordinate result in the proof of Theorem 1 in~\cite{DBLP:conf/crypto/KaplanLLN16}.
\begin{claim}
Let $h : \{0,1\}^{n} \rightarrow \{0,1\}^\ell$ be a function and $\ket
{\Phi} := \sum_{u,x} (-1)^{u \cdot x} \ket{u}\ket{h(x)}$.
Then, for each $t \in \{0,1\}^{n} \setminus \{0^{n}\}$,
\begin{equation}
\Pr\left[u \gets (\text{measure } \ket{\Phi}) : t \perp u \right] = \frac{1}{2}\left( 1 + \Pr_{x \gets \{0,1\}^{n}}\left[ h(x \oplus t) = h(x) \right] \right)
\end{equation}
holds.
\end{claim}
Applying this claim with $(f_i \oplus g)$ and $\ket{\Psi^{(i)}}$ instead of $h$ and $\ket{\Phi}$, from \eqref{eq:Psi1ProbBad} it follows that $p_{\bad}^{(i)}$ is upper bounded by
\begin{align}
2^n   \left(\frac{1}{2}\left( 1 + \max_{t \in \{0,1\}^{n} \setminus \{0^{n}\}} \Pr_{x \gets \{0,1\}^{n}}\left[ (f_i \oplus g)(x \oplus t)= (f_i \oplus g)(x) \right] \right) \right)^{cn}.
\end{align}

In addition, by the condition \eqref{eq:SimonErrorHScondition},
\begin{align}
\max_{ \substack{t \in \{0,1\}^{n} \setminus \{0^{n}\} \\ f_i \in F_0}} \Pr_{x \gets \{0,1\}^{n}}\left[ (f_i \oplus g)(x \oplus t)= (f_i \oplus g)(x) \right] \leq \epsilon
\end{align}
holds for all $f_i \in F_0$.
Thus
\begin{align}
p_\bad^{(i)} \leq 2^n\left( \frac{1 + \epsilon}{2} \right)^{cn}. \label{eq:SimonPbadUpperFinal}
\end{align}
follows.

From \eqref{eq:SimonDeltaBound} and \eqref{eq:SimonPbadUpperFinal},
\begin{equation}
\| \ket{\delta_b} \|^2 \leq 2^{n+1}\left(\frac{1 + \epsilon}{2}\right)^{cn}
\end{equation}
follows, which completes the proof.
\qed
\end{proof}

\subsection{Error Propagation Analyses for QAA}\label{sec:ErrorQAA}
Here we analyze how the error of testing procedures affects success probability and complexity of the entire Grover search.
For generality, we analyze the quantum amplitude amplification (QAA) technique by Brassard \emph{et al.}~\cite{brassard2002quantum} rather than the original Grover search.
We first review the original QAA technique, and then give error analyses in our setting, explaining the difference between our setting and the original one.

\subsubsection{The Original Quantum Amplitude Amplification}
Let $\mathcal{A}$ be a unitary operator that acts on $n$-qubit states and $\chi : \{0,1\}^n \rightarrow \{0,1\}$ be a function.
Suppose that we obtain $x$ such that $\chi(x)=1$ when we measure the state $\mathcal{A}\ket{0^n}$ with a probability $a > 0$.
We say that $x \in \{0,1\}^n$ is \emph{good} if $\chi(x)=1$ and $x$ is \emph{bad} otherwise.
In addition, we say that a vector is good (resp. bad) if it is in the space spanned by $\{ \ket{x} \}_{x : \text{good}}$ (resp., $\{ \ket{x} \}_{x : \text{bad}}$).
Our goal is to amplify the probability $a$ and get ``good'' $x \in \{0,1\}^n$ such that $\chi(x)=1$ with a high probability.

Let $\mathcal{S}_\chi$ be the unitary operator defined as
\begin{equation}
\mathcal{S}_\chi\ket{x} :=
\begin{cases}
-\ket{x}  \text{ if } \chi(x)= 1, \\
\phantom{-}\ket{x}  \text{ if } \chi(x) = 0.
\end{cases}
\end{equation}
Below we call $\mathcal{S}_\chi$ \emph{checking procedure} since it checks whether $\chi(x)=1$ and changes the phase accordingly.
In addition, let $\mathcal{S}_0$ denote $\mathcal{S}_{\chi_0}$, where $\chi_0$ is the function such that $\chi_0(x)=1$ if and only if $x = 0^n$.
Define a unitary operator $Q=Q(\mathcal{A},\chi)$ by $Q := - \mathcal{A}\mathcal{S}_0\mathcal{A}^{-1}\mathcal{S}_\chi$.
Then the following proposition holds.
\begin{proposition}[Theorem 2 in \cite{brassard2002quantum}.]\label{prop:QAAknowinga}
Set $r := \lfloor \pi / 4\theta_a \rfloor$, where $\theta_a$ is the parameter such that $\sin^2 \theta_a = a$ and $0 < \theta_a \leq \pi/2$ (note that $r \approx \sqrt{1/a}$ holds if a is sufficiently small).
When we measure $Q^r \mathcal{A} \ket{0^n}$, we obtain $x \in \{0,1\}^n$ such that $\chi(x) =1$ with a probability at least $\max{(1-a,a)}$.
\end{proposition}
Here we give a rough overview of a proof of the above proposition.
Let us put $\ket{\phi} := \mathcal{A}\ket{0^n}$, and decompose $\ket{\phi}$ as $\ket{\phi} = \cos\theta_a\ket{\phi_0} + \sin\theta_a\ket{\phi_1}$, where $\ket{\phi_0}$ and $\ket{\phi_1}$ are good and bad vectors, respectively.
Note that $\ket{\phi_0}$ and $\ket{\phi_1}$ are uniquely determined from $\ket{\phi}$.
Then we can show the following lemma, which claims that $Q$ can be regarded as a rotation matrix on the space spanned by $\ket{\phi_0}$ and $\ket{\phi_1}$.
\begin{lemma}[Lemma 1 in~\cite{brassard2002quantum}]\label{lem:QisRotationMatrix}
It holds that
\begin{align}
Q\ket{\phi_1} &= \cos2\theta_a \ket{\phi_1} - \sin2\theta_a \ket{\phi_0}, \nonumber \\
Q\ket{\phi_0} &= \sin2\theta_a \ket{\phi_1} +\cos2\theta_a \ket{\phi_0}. \nonumber
\end{align}
\end{lemma}
From the above lemma, we have that
\begin{equation}
Q^r\mathcal{A}\ket{0^n} = \sin((2r+1)\theta_a) \ket{\phi_1} + \cos((2r+1)\theta_a) \ket{\phi_0},
\end{equation}
which implies that we obtain $x \in \{0,1\}^n$ such that $\chi(x)=1$ since $(2r+1)\theta_a \approx \pi/2$.


\subsubsection{Our Setting: Uncertain Checking Procedures}
Unlike the original technique, in our setting we consider the situation that the checking procedure can be done only with some errors.
That is, we have access to a unitary operator $\widehat{\mathcal{S}'_\chi}$ such that
\begin{align}
\widehat{\mathcal{S}'_\chi} \ket{\phi_1}\ket{\psi}\ket{b} &= \ket{\phi_1}\ket{\psi}\ket{b \oplus 1} + \ket{\delta_{1,b}},  \\
\widehat{\mathcal{S}'_\chi} \ket{\phi_0}\ket{\psi}\ket{b} &= \ket{\phi_0}\ket{\psi}\ket{b} + \ket{\delta_{0,b}},
\end{align}
where $\ket{\psi}$ is an auxiliary data and $b \in \{0,1\}$, and there exists an $\epsilon \geq 0$ such that $\|\delta_{a,b}\| \leq \epsilon$ holds for $a,b \in \{0,1\}$.
(In our attack, the operator $\widehat{\mathcal{S}'_\chi}$ is denoted by $\sf test$.)

Note that, given such $\widehat{\mathcal{S}'_\chi}$, by using an ancilla qubit we can implement a unitary operator $\mathcal{S}'_\chi$ such that
\begin{align}
{\mathcal{S}'_\chi} \ket{\phi_1}\ket{\psi} &= -\ket{\phi_1}\ket{\psi} + \ket{\delta_1}, \\
{\mathcal{S}'_\chi} \ket{\phi_0}\ket{\psi} &= \phantom{-}\ket{\phi_0}\ket{\psi} + \ket{\delta_0},
\end{align}
where $\| \ket{\delta_0} \|, \| \ket{\delta_1}\| \leq 2\epsilon$ holds, since
\begin{equation}
\widehat{\mathcal{S}'_\chi} \ket{\phi_a}\ket{\psi}\ket{-} = (-1)^a\ket{\phi_a}\ket{\psi}\ket{-} + \frac{1}{\sqrt{2}}(\ket{\delta_{a,0}} - \ket{\delta_{a,1}})
\end{equation}
holds for $a \in \{0,1\}$, here $\ket{-} := \frac{1}{\sqrt{2}}(\ket{0} - \ket{1})$.

Let us define $Q' := -((\mathcal{A}\mathcal{S}_0 \mathcal{A}^{-1}) \otimes I) \mathcal{S}'_\chi$.
Then the following lemma holds.
\begin{lemma}\label{lem:QisRotationMatrixError}
There exist vectors $\ket{\delta'_0},\ket{\delta'_1}$ such that $\|\ket{\delta'_0}\|,\|\ket{\delta'_1}\| \leq 2\epsilon$ and
\begin{align}
Q'\ket{\phi_1}\ket{\psi} &= \left( \cos2\theta_a \ket{\phi_1} - \sin2\theta_a \ket{\phi_0}\right) \ket{\psi} + \ket{\delta'_1}, \nonumber \\
Q'\ket{\phi_0} \ket{\psi} &= \left( \sin2\theta_a \ket{\phi_1} +\cos2\theta_a \ket{\phi_0} \right) \ket{\psi} + \ket{\delta'_0} \nonumber
\end{align}
hold.
\end{lemma}
\begin{proof}
We have that
\begin{equation}
\|Q'\ket{\phi_1}\ket{\psi} - (Q\otimes I)\ket{\phi_1}\ket{\psi} \| = \|\mathcal{S}'_\chi\ket{\phi_1}\ket{\psi} - (\mathcal{S}_\chi \otimes I)\ket{\phi_1} \ket{\psi} \| = \|\delta_1\| \leq 2\epsilon
\end{equation}
holds.
Therefore the first equality follows from Lemma~\ref{lem:QisRotationMatrix}.
We can show the second equality in the same way.
\qed
\end{proof}
By using this lemma we can show the following proposition.
\begin{proposition}[QAA with uncertain checking procedures.]\label{prop:QAAuncertainKnowinga}
Let $\epsilon$ be the error of $\widehat{\mathcal{S}'_\chi}$ described above.
Then we have that
\begin{align}
&\left| \Pr\left[ x \gets \left(\text{\rm measure } {Q'}^j \mathcal{A}\ket{0^n}\ket{\psi} \right) : \chi(x) = 1 \right] \right. \nonumber \\
&\qquad \quad \left. - \Pr\left[ x \gets \left(\text{\rm measure } {Q}^j \mathcal{A}\ket{0^n}\right) : \chi(x) = 1 \right] \right| \leq 4j\epsilon
\end{align}
holds for all $x \in \{0,1\}^n$.
In particular, for $r := \lfloor \pi / 4\theta_a \rfloor$, we obtain $x \in \{0,1\}^n$ such that $\chi(x) =1$ with a probability at least
$\max{(1-a,a)} - 4r\epsilon $ when we measure ${Q'}^r \mathcal{A} \ket{0^n}\ket{\psi}$.
\end{proposition}
\begin{proof}
From Lemma~\ref{lem:QisRotationMatrixError}, it follows that there exists a vector $\ket{\delta^{(j)}}$ such that $\| \ket{\delta^{(j)}} \| \leq 4j\epsilon$ and
\begin{align}
{Q'}^j\mathcal{A} \ket{0^n}\ket{\psi}
& = \left( \sin((2j+1)\theta_a) \ket{\phi_1} + \cos((2j+1)\theta_a) \ket{\phi_0} \right) \otimes \ket{\psi} + \ket{\delta^{(j)}} \nonumber \\
& = (Q^j \mathcal{A} \ket{0^n})\ket{\psi} + \ket{\delta^{(j)}} \label{eq:QAAerrorQjPrime}
\end{align}
holds.
Let us put
\begin{equation}
p' := \Pr\left[ x \gets \left(\text{\rm measure } {Q'}^j \mathcal{A}\ket{0^n}\ket{\xi} \right) : \chi(x) = 1 \right]
\end{equation}
and
\begin{equation}
p := \Pr\left[ x \gets \left(\text{\rm measure } {Q}^j \mathcal{A}\ket{0^n}\right) : \chi(x) = 1 \right].
\end{equation}
Then, since 
\begin{equation}
p = \Pr\left[ x \gets \left(\text{\rm measure } ({Q}^j \mathcal{A}\ket{0^n}) \ket{\psi} \right) : \chi(x) = 1 \right]
\end{equation}
holds, we have that
\begin{align}
\left| p - p' \right| &\leq \|({Q}^j \mathcal{A}\ket{0^n}) \ket{\psi} - {Q'}^j \mathcal{A}\ket{0^n}\ket{\psi}\| \nonumber \\
&= \| \delta^{(j)} \| \leq 4j\epsilon,
\end{align}
which completes the proof.
\qed
\end{proof}


\subsection{Finishing the Proofs}\label{sec:ErrorFin}
Here we complete the proofs for Propositions~\ref{prop:AlgQ2} and \ref{prop:AlgQ1}.
Since the number of queries required to make the state $\ket{\psi_g}$ is obviously $cn$ in the Q2 model and $2^n$ in the Q1 model, we give proofs for the statements on the offline computation (i.e., complexities and success probabilities for the procedures excluding the ones to prepare the state $\ket{\psi_g}$).
Hence it suffices to show the following lemma to prove Propositions~\ref{prop:AlgQ2} and \ref{prop:AlgQ1}.
\begin{lemma}\label{lemma:AlgQ1Q2common}
Suppose that $m$ is in $\bigO{n}$.
Let $c$ be a sufficiently large constant,\footnote{Strictly speaking, it is sufficient if $c$ satisfies $4\lfloor \pi / 4 \theta \rfloor 2^{(n+1)/2}(3/4)^{cn/2}< 1/2$, where $\theta$ is a positive value such that $0 < \theta < \pi/2$ and $\sin^2 \theta = 1/2^m$.} and let $i_0 \in \{0,1\}^m$ be the good element such that $g \oplus f_i$ is periodic.
Assume that 
\begin{equation}\label{eq:OffLemmaEpsilon}
\max_{ \substack{i \gets \{0,1\}^m \setminus \{i_0\} \\ t \in \{0,1\}^{n} \setminus \{0^{n}\}}} \Pr_{ x \gets \{0,1\}^{n}}\left[ (f_i \oplus g)(x \oplus t)= (f_i \oplus g)(x) \right] \leq \frac{1}{2}
\end{equation}
holds, and the quantum state $\ket{\psi_g}$ is given.
Then, the offline phase of $\AlgQone$ and $\AlgQtwo$ finds a good $i \in \{0,1\}^m$ with a probability in $\Theta(1)$ by making $\bigO{n2^{m/2}}$ quantum queries to $F$.
In addition, the offline computation is done in time $O((n^3+nT_F)2^{m/2})$, where is $T_F$ the time required to evaluate $F$ once.
\end{lemma}
\begin{proof}
Let $\ket{\phi_1} := \ket{i_0}\ket{\psi_{f_{i_0} \oplus g}}$ and $\ket{\phi_0} := \sum_{i \in \{0,1\}^m \setminus \{i_0\} } \frac{1}{\sqrt{2^m-1}}\ket{i} \ket{\psi_{f_i\oplus g}}$.
Then, by the condition \eqref{eq:OffLemmaEpsilon} and Lemma~\ref{lem:SimonErrorBound}, we have that
\begin{align}
{\sf test} \ket{\phi_1} &= \ket{\phi_1}, \text{ and } \nonumber \\
{\sf test} \ket{\phi_0} &= \ket{\phi_0} + \ket{\delta},
\end{align}
where $\ket{\delta}$ is a vector such that $\| \delta \| < 2^{(n+1)/2} (3/4)^{cn/2}$.

Now, the offline phases of $\AlgQone$ and $\AlgQtwo$ correspond to the (modified) quantum amplitude amplification of Proposition~\ref{prop:QAAuncertainKnowinga} with $\mathcal{A} = H^{\otimes m}$, $\widehat{\mathcal{S}'_\chi} = {\sf test}$, $\epsilon = 2^{(n+1)/2} (3/4)^{cn/2}$, and $a = 1/2^m$.
Therefore, by applying the Grover's iteration (i.e., $Q'$) $\bigO{2^{m/2}}$ times, we can obtain the good index $i_0$
with a probability at least $1 - 2^{-m/2} - O(2^{m/2} 2^{(n+1)/2} (3/4)^{cn/2})$.
Since we are assuming that $c$ is a sufficiently large constant and $m = \bigO{n}$, $\bigO{2^{m/2} 2^{(n+1)/2} (3/4)^{cn/2}) } o(1)$ holds.
Thus we can obtain the good index $i_0$ with a probability in $\Theta(1)$.

Since each Grover's iteration (i.e., $Q'$) makes $\bigO{cn}=\bigO{n}$ queries to $F$, the total number of quantum queries to $F$ is $\bigO{n 2^{m/2}}$.
In addition, each Grover's iteration runs in time $\bigO{n^3 + nT_F}$ since computing the dimension of the space spanned by $cn$ can be done in time $\bigO{n^3}$.
Thus the offline phase runs in time $\bigO{(n^3 + nT_F)2^{m/2}}$.
\qed
\end{proof}

%% file: simoncombi.bbl
\begin{thebibliography}{10}
\providecommand{\url}[1]{\texttt{#1}}
\providecommand{\urlprefix}{URL }

\bibitem{DBLP:conf/crypto/AlbrechtDKLPY14}
Albrecht, M.R., Driessen, B., Kavun, E.B., Leander, G., Paar, C., Yal{\c{c}}in,
  T.: Block ciphers - focus on the linear layer (feat. {PRIDE)}. In: {CRYPTO}
  {(2)}. Lecture Notes in Computer Science, vol. 8616, pp. 57--76. Springer
  (2014)

\bibitem{DBLP:journals/tosc/BertoniDHPAK17}
Bertoni, G., Daemen, J., Hoffert, S., Peeters, M., Assche, G.V., Keer, R.V.:
  Farfalle: parallel permutation-based cryptography. {IACR} Trans. Symmetric
  Cryptol.  2017(4),  1--38 (2017),
  \url{https://tosc.iacr.org/index.php/ToSC/article/view/801}

\bibitem{DBLP:conf/fse/BiryukovW99}
Biryukov, A., Wagner, D.A.: Slide attacks. In: {FSE}. Lecture Notes in Computer
  Science, vol. 1636, pp. 245--259. Springer (1999)

\bibitem{journals/iacr/Bonnetain17}
Bonnetain, X.: Quantum key-recovery on full {AEZ}. In: Selected Areas in
  Cryptography - {SAC} 2017. Lecture Notes in Computer Science, vol. 10719, pp.
  394--406. Springer (2018)

\bibitem{DBLP:conf/asiacrypt/BonnetainN18}
Bonnetain, X., Naya{-}Plasencia, M.: Hidden shift quantum cryptanalysis and
  implications. In: {ASIACRYPT} 2018. Lecture Notes in Computer Science, vol.
  11272, pp. 560--592. Springer (2018)

\bibitem{DBLP:journals/iacr/BonnetainNS18}
Bonnetain, X., Naya-Plasencia, M., Schrottenloher, A.: On quantum slide
  attacks. In: Selected Areas in Cryptography - {SAC} 2019. Lecture Notes in
  Computer Science, Springer (2020)

\bibitem{DBLP:conf/asiacrypt/BorghoffCGKKKLNPRRTY12}
Borghoff, J., Canteaut, A., G{\"{u}}neysu, T., Kavun, E.B., Knezevic, M.,
  Knudsen, L.R., Leander, G., Nikov, V., Paar, C., Rechberger, C., Rombouts,
  P., Thomsen, S.S., Yal{\c{c}}in, T.: {PRINCE} - {A} low-latency block cipher
  for pervasive computing applications - extended abstract. In: {ASIACRYPT}.
  Lecture Notes in Computer Science, vol. 7658, pp. 208--225. Springer (2012)

\bibitem{brassard2002quantum}
Brassard, G., Hoyer, P., Mosca, M., Tapp, A.: {Quantum amplitude amplification
  and estimation}. Contemporary Mathematics  305,  53--74 (2002)

\bibitem{DBLP:conf/latin/BrassardHT98}
Brassard, G., H{\o}yer, P., Tapp, A.: Quantum cryptanalysis of hash and
  claw-free functions. In: Lucchesi, C.L., Moura, A.V. (eds.) {LATIN} '98:
  Theoretical Informatics, Third Latin American Symposium, Campinas, Brazil,
  April, 20-24, 1998, Proceedings. Lecture Notes in Computer Science, vol.
  1380, pp. 163--169. Springer (1998), \url{https://doi.org/10.1007/BFb0054319}

\bibitem{saturnin}
Canteaut, A., Duval, S., Leurent, G., Naya-Plasencia, M., Perrin, L., Pornin,
  T., Schrottenloher, A.: {Saturnin:} a suite of lightweight symmetric
  algorithms for post-quantum security (2019),
  \url{https://project.inria.fr/saturnin/files/2019/05/SATURNIN-spec.pdf}

\bibitem{add:CarChaZin98}
Carlet, C., Charpin, P., Zinoviev, V.: Codes, bent functions and permutations
  suitable for {DES}-like cryptosystems. Designs, Codes and Cryptography
  15(2),  125--156 (1998)

\bibitem{DBLP:conf/asiacrypt/ChaillouxNS17}
Chailloux, A., Naya{-}Plasencia, M., Schrottenloher, A.: An efficient quantum
  collision search algorithm and implications on symmetric cryptography. In:
  {ASIACRYPT} {(2)}. Lecture Notes in Computer Science, vol. 10625, pp.
  211--240. Springer (2017)

\bibitem{DBLP:journals/tches/ChakrabortiDNY18}
Chakraborti, A., Datta, N., Nandi, M., Yasuda, K.: Beetle family of lightweight
  and secure authenticated encryption ciphers. {IACR} Trans. Cryptogr. Hardw.
  Embed. Syst.  2018(2),  218--241 (2018),
  \url{https://doi.org/10.13154/tches.v2018.i2.218-241}

\bibitem{DBLP:journals/tosc/CrowleyB18}
Crowley, P., Biggers, E.: Adiantum: length-preserving encryption for
  entry-level processors. {IACR} Trans. Symmetric Cryptol.  2018(4),  39--61
  (2018), \url{https://doi.org/10.13154/tosc.v2018.i4.39-61}

\bibitem{DBLP:conf/asiacrypt/Daemen91}
Daemen, J.: Limitations of the even-mansour construction. In: {ASIACRYPT} 1991.
  Lecture Notes in Computer Science, vol. 739, pp. 495--498. Springer (1991)

\bibitem{DBLP:journals/tosc/DaemenHAK18}
Daemen, J., Hoffert, S., Assche, G.V., Keer, R.V.: The design of xoodoo and
  xoofff. {IACR} Trans. Symmetric Cryptol.  2018(4),  1--38 (2018),
  \url{https://doi.org/10.13154/tosc.v2018.i4.1-38}

\bibitem{DBLP:conf/eurocrypt/Dinur15}
Dinur, I.: Cryptanalytic time-memory-data tradeoffs for fx-constructions with
  applications to {PRINCE} and {PRIDE}. In: {EUROCRYPT} 2015. Lecture Notes in
  Computer Science, vol. 9056, pp. 231--253. Springer (2015)

\bibitem{DBLP:conf/asiacrypt/DinurDKS14}
Dinur, I., Dunkelman, O., Keller, N., Shamir, A.: Cryptanalysis of iterated
  even-mansour schemes with two keys. In: {ASIACRYPT} 2014. Lecture Notes in
  Computer Science, vol. 8873, pp. 439--457. Springer (2014)

\bibitem{DBLP:journals/joc/EvenM97}
Even, S., Mansour, Y.: A construction of a cipher from a single pseudorandom
  permutation. J. Cryptology  10(3),  151--162 (1997),
  \url{http://dx.doi.org/10.1007/s001459900025}

\bibitem{DBLP:phd/dnb/Gagliardoni17}
Gagliardoni, T.: Quantum Security of Cryptographic Primitives. Ph.D. thesis,
  Darmstadt University of Technology, Germany (2017),
  \url{http://tuprints.ulb.tu-darmstadt.de/6019/}

\bibitem{DBLP:conf/pqcrypto/GrasslLRS16}
Grassl, M., Langenberg, B., Roetteler, M., Steinwandt, R.: Applying grover's
  algorithm to {AES:} quantum resource estimates. In: PQCrypto. Lecture Notes
  in Computer Science, vol. 9606, pp. 29--43. Springer (2016)

\bibitem{DBLP:conf/stoc/Grover96}
Grover, L.K.: {A Fast Quantum Mechanical Algorithm for Database Search}. In:
  Miller, G.L. (ed.) {Proceedings of the Twenty-Eighth Annual {ACM} Symposium
  on the Theory of Computing, Philadelphia, Pennsylvania, USA, May 22-24,
  1996}. pp. 212--219. {ACM} (1996),
  \url{http://doi.acm.org/10.1145/237814.237866}

\bibitem{DBLP:conf/ctrsa/HosoyamadaS18}
Hosoyamada, A., Sasaki, Y.: Cryptanalysis against symmetric-key schemes with
  online classical queries and offline quantum computations. In: {CT-RSA}.
  Lecture Notes in Computer Science, vol. 10808, pp. 198--218. Springer (2018)

\bibitem{DBLP:conf/crypto/KaplanLLN16}
Kaplan, M., Leurent, G., Leverrier, A., Naya{-}Plasencia, M.: Breaking
  symmetric cryptosystems using quantum period finding. In: {CRYPTO} {(2)}.
  Lecture Notes in Computer Science, vol. 9815, pp. 207--237. Springer (2016)

\bibitem{DBLP:journals/tosc/KaplanLLN16}
Kaplan, M., Leurent, G., Leverrier, A., Naya{-}Plasencia, M.: Quantum
  differential and linear cryptanalysis. {IACR} Trans. Symmetric Cryptol.
  2016(1),  71--94 (2016),
  \url{http://tosc.iacr.org/index.php/ToSC/article/view/536}

\bibitem{DBLP:conf/crypto/KilianR96}
Kilian, J., Rogaway, P.: How to protect {DES} against exhaustive key search.
  In: {CRYPTO}. Lecture Notes in Computer Science, vol. 1109, pp. 252--267.
  Springer (1996)

\bibitem{DBLP:journals/siamcomp/Kuperberg05}
Kuperberg, G.: A subexponential-time quantum algorithm for the dihedral hidden
  subgroup problem. {SIAM} J. Comput.  35(1),  170--188 (2005),
  \url{https://doi.org/10.1137/S0097539703436345}

\bibitem{DBLP:conf/tqc/Kuperberg13}
Kuperberg, G.: Another subexponential-time quantum algorithm for the dihedral
  hidden subgroup problem. In: {TQC} 2013. LIPIcs, vol.~22, pp. 20--34. Schloss
  Dagstuhl - Leibniz-Zentrum fuer Informatik (2013)

\bibitem{DBLP:conf/isit/KuwakadoM10}
Kuwakado, H., Morii, M.: Quantum distinguisher between the 3-round feistel
  cipher and the random permutation. In: {IEEE} International Symposium on
  Information Theory, {ISIT} 2010, Proceedings. pp. 2682--2685. {IEEE} (2010)

\bibitem{DBLP:conf/isita/KuwakadoM12}
Kuwakado, H., Morii, M.: Security on the quantum-type even-mansour cipher. In:
  Proceedings of the International Symposium on Information Theory and its
  Applications, {ISITA} 2012. pp. 312--316. {IEEE} (2012)

\bibitem{DBLP:conf/asiacrypt/Leander017}
Leander, G., May, A.: {Grover Meets Simon - Quantumly Attacking the
  FX-construction}. In: {ASIACRYPT} 2017. {Lecture Notes in Computer Science},
  vol. 10625, pp. 161--178. Springer (2017)

\bibitem{DBLP:journals/ieeesp/Martin10}
Martin, L.: {XTS:} {A} mode of {AES} for encrypting hard disks. {IEEE} Security
  {\&} Privacy  8(3),  68--69 (2010),
  \url{https://doi.org/10.1109/MSP.2010.111}

\bibitem{DBLP:conf/sacrypt/MouhaMHWPV14}
Mouha, N., Mennink, B., Herrewege, A.V., Watanabe, D., Preneel, B.,
  Verbauwhede, I.: Chaskey: An efficient {MAC} algorithm for 32-bit
  microcontrollers. In: Selected Areas in Cryptography. Lecture Notes in
  Computer Science, vol. 8781, pp. 306--323. Springer (2014)

\bibitem{NAP25196}
{National Academies of Sciences{,} Engineering{,} and Medicine}: Quantum
  Computing: Progress and Prospects. The National Academies Press,
  Washington{,} DC (2018),
  \url{https://www.nap.edu/catalog/25196/quantum-computing-progress-and-prospects}

\bibitem{nistcall}
{National Institute of Standards and Technlology}: Submission requirements and
  evaluation criteria for the post-quantum cryptography standardization process
  (2016), \url{https://csrc.nist.gov/CSRC/media/
  Projects/Post-Quantum-Cryptography/documents/call-for-proposals-final-dec-2016.pdf}

\bibitem{nielsen2002quantum}
Nielsen, M.A., Chuang, I.: Quantum computation and quantum information. AAPT
  (2002)

\bibitem{DBLP:journals/ipl/RottelerS15}
R{\"{o}}tteler, M., Steinwandt, R.: A note on quantum related-key attacks. Inf.
  Process. Lett.  115(1),  40--44 (2015),
  \url{https://doi.org/10.1016/j.ipl.2014.08.009}

\bibitem{Minalpher}
Sasaki, Y., Todo, Y., Aoki, K., Naito, Y., Sugawara, T., Murakami, Y., Matsui,
  M., Hirose, S.: {Minalpher} v1.1. {CAESAR} competition.  (2015),
  \url{https://competitions.cr.yp.to/round2/minalpherv11.pdf}

\bibitem{DBLP:conf/focs/Shor94}
Shor, P.W.: Algorithms for quantum computation: Discrete logarithms and
  factoring. In: 35th Annual Symposium on Foundations of Computer Science. pp.
  124--134. {IEEE} Computer Society (1994)

\bibitem{DBLP:conf/focs/Simon94}
Simon, D.R.: {On the Power of Quantum Computation}. In: {35th Annual Symposium
  on Foundations of Computer Science}. pp. 116--123 (1994)

\bibitem{DBLP:journals/cryptologia/WinternitzH87}
Winternitz, R.S., Hellman, M.E.: Chosen-key attacks on a block cipher.
  Cryptologia  11(1),  16--20 (1987),
  \url{https://doi.org/10.1080/0161-118791861749}

\end{thebibliography}
